\documentclass[journal]{IEEEtran}
\makeatletter

\newcommand{\Rmnum}[1]{\expandafter\@slowromancap\romannumeral #1@}
\makeatother
\usepackage{array}
\usepackage{multirow}
\usepackage{diagbox}
\usepackage{amsthm,amsmath,amsfonts,syntonly,setspace}
\usepackage{amssymb}
\usepackage{epsfig} % for postscript graphics files
\usepackage{epstopdf}
\usepackage{epsf,graphics,graphicx}
\usepackage{cite,bm,color,url,textcomp}
\usepackage{float}
\usepackage{algorithm}
\usepackage{algorithmic}
\usepackage{mathrsfs}
\usepackage{arydshln}
\usepackage{stfloats}
\usepackage{subfigure}
\usepackage{lipsum}
\usepackage{cuted}%%\stripsep-3pt
\newcommand\blfootnote[1]{%
  \begingroup
  \renewcommand\thefootnote{}\footnote{#1}%
  \addtocounter{footnote}{-1}%
  \endgroup
}

\newtheorem{proposition}{Proposition}
\begin{document}
%--------------------------------------------
%command file for *.tex files
%--------------------------------------------
\newcommand{\fs}{\hspace{0.07in}}
\newcommand{\bs}{\hspace{-0.1in}}
\newcommand{\re}{{\rm Re} \, }
\newcommand{\e}{{\rm E} \, }
\newcommand{\p}{{\rm P} \, }
\newcommand{\cn}{{\cal CN} \, }
\newcommand{\n}{{\cal N} \, }
\newcommand{\ba}{\begin{array}}
\newcommand{\ea}{\end{array}}
\newcommand{\be}{\begin{displaymath}}
\newcommand{\ee}{\end{displaymath}}
\newcommand{\ben}{\begin{equation}}
\newcommand{\een}{\end{equation}}
\newcommand{\bena}{\begin{eqnarray}}
\newcommand{\eena}{\end{eqnarray}}
\newcommand{\beqa}{\begin{eqnarray*}}
\newcommand{\enqa}{\end{eqnarray*}}
\newcommand{\f}{\frac}
\newcommand{\bc}{\begin{center}}
\newcommand{\ec}{\end{center}}
\newcommand{\bi}{\begin{itemize}}
\newcommand{\ei}{\end{itemize}}
\newcommand{\benu}{\begin{enumerate}}
\newcommand{\eenu}{\end{enumerate}}
\newcommand{\bdes}{\begin{description}}
\newcommand{\edes}{\end{description}}
\newcommand{\bt}{\begin{tabular}}
\newcommand{\et}{\end{tabular}}
\newcommand{\vs}{\vspace}
\newcommand{\hs}{\hspace}
\newcommand{\sort}{\rm sort \,}

\newcommand \thetabf{{\mbox{\boldmath$\theta$\unboldmath}}}
\newcommand{\Phibf}{\mbox{${\bf \Phi}$}}
\newcommand{\Psibf}{\mbox{${\bf \Psi}$}}
\newcommand \alphabf{\mbox{\boldmath$\alpha$\unboldmath}}
\newcommand \betabf{\mbox{\boldmath$\beta$\unboldmath}}
\newcommand \gammabf{\mbox{\boldmath$\gamma$\unboldmath}}
\newcommand \deltabf{\mbox{\boldmath$\delta$\unboldmath}}
\newcommand \epsilonbf{\mbox{\boldmath$\epsilon$\unboldmath}}
\newcommand \zetabf{\mbox{\boldmath$\zeta$\unboldmath}}
\newcommand \etabf{\mbox{\boldmath$\eta$\unboldmath}}
\newcommand \iotabf{\mbox{\boldmath$\iota$\unboldmath}}
\newcommand \kappabf{\mbox{\boldmath$\kappa$\unboldmath}}
\newcommand \lambdabf{\mbox{\boldmath$\lambda$\unboldmath}}
\newcommand \mubf{\mbox{\boldmath$\mu$\unboldmath}}
\newcommand \nubf{\mbox{\boldmath$\nu$\unboldmath}}
\newcommand \xibf{\mbox{\boldmath$\xi$\unboldmath}}
\newcommand \pibf{\mbox{\boldmath$\pi$\unboldmath}}
\newcommand \rhobf{\mbox{\boldmath$\rho$\unboldmath}}
\newcommand \sigmabf{\mbox{\boldmath$\sigma$\unboldmath}}
\newcommand \taubf{\mbox{\boldmath$\tau$\unboldmath}}
\newcommand \upsilonbf{\mbox{\boldmath$\upsilon$\unboldmath}}
\newcommand \phibf{\mbox{\boldmath$\phi$\unboldmath}}
\newcommand \varphibf{\mbox{\boldmath$\varphi$\unboldmath}}
\newcommand \chibf{\mbox{\boldmath$\chi$\unboldmath}}
\newcommand \psibf{\mbox{\boldmath$\psi$\unboldmath}}
\newcommand \omegabf{\mbox{\boldmath$\omega$\unboldmath}}
\newcommand \Sigmabf{\hbox{$\bf \Sigma$}}
\newcommand \Upsilonbf{\hbox{$\bf \Upsilon$}}
\newcommand \Omegabf{\hbox{$\bf \Omega$}}
\newcommand \Deltabf{\hbox{$\bf \Delta$}}
\newcommand \Gammabf{\hbox{$\bf \Gamma$}}
\newcommand \Thetabf{\hbox{$\bf \Theta$}}
\newcommand \Lambdabf{\hbox{$\bf \Lambda$}}
\newcommand \Xibf{\hbox{\bf$\Xi$}}
\newcommand \Pibf{\hbox{\bf$\Pi$}}
\newcommand \abf{{\bf a}}
\newcommand \bbf{{\bf b}}
\newcommand \cbf{{\bf c}}
\newcommand \dbf{{\bf d}}
\newcommand \ebf{{\bf e}}
\newcommand \fbf{{\bf f}}
\newcommand \gbf{{\bf g}}
\newcommand \hbf{{\bf h}}
\newcommand \ibf{{\bf i}}
\newcommand \jbf{{\bf j}}
\newcommand \kbf{{\bf k}}
\newcommand \lbf{{\bf l}}
\newcommand \mbf{{\bf m}}
\newcommand \nbf{{\bf n}}
\newcommand \obf{{\bf o}}
\newcommand \pbf{{\bf p}}
\newcommand \qbf{{\bf q}}
\newcommand \rbf{{\bf r}}
\newcommand \sbf{{\bf s}}
\newcommand \tbf{{\bf t}}
\newcommand \ubf{{\bf u}}
\newcommand \vbf{{\bf v}}
\newcommand \wbf{{\bf w}}
\newcommand \xbf{{\bf x}}
\newcommand \ybf{{\bf y}}
\newcommand \zbf{{\bf z}}
\newcommand \rbfa{{\bf r}}
\newcommand \xbfa{{\bf x}}
\newcommand \ybfa{{\bf y}}
\newcommand \Abf{{\bf A}}
\newcommand \Bbf{{\bf B}}
\newcommand \Cbf{{\bf C}}
\newcommand \Dbf{{\bf D}}
\newcommand \Ebf{{\bf E}}
\newcommand \Fbf{{\bf F}}
\newcommand \Gbf{{\bf G}}
\newcommand \Hbf{{\bf H}}
\newcommand \Ibf{{\bf I}}
\newcommand \Jbf{{\bf J}}
\newcommand \Kbf{{\bf K}}
\newcommand \Lbf{{\bf L}}
\newcommand \Mbf{{\bf M}}
\newcommand \Nbf{{\bf N}}
\newcommand \Obf{{\bf O}}
\newcommand \Pbf{{\bf P}}
\newcommand \Qbf{{\bf Q}}
\newcommand \Rbf{{\bf R}}
\newcommand \Sbf{{\bf S}}
\newcommand \Tbf{{\bf T}}
\newcommand \Ubf{{\bf U}}
\newcommand \Vbf{{\bf V}}
\newcommand \Wbf{{\bf W}}
\newcommand \Xbf{{\bf X}}
\newcommand \Ybf{{\bf Y}}
\newcommand \Zbf{{\bf Z}}
\newcommand \Omegabbf{{\bf \Omega}}
\newcommand \Rssbf{{\bf R_{ss}}}
\newcommand \Ryybf{{\bf R_{yy}}}
\newcommand \Cset{{\cal C}}
\newcommand \Rset{{\cal R}}
\newcommand \Zset{{\cal Z}}
\newcommand{\otheta}{\stackrel{\circ}{\theta}}
\newcommand{\defeq}{\stackrel{\bigtriangleup}{=}}
\newcommand{\oabf}{{\bf \breve{a}}}
\newcommand{\odbf}{{\bf \breve{d}}}
\newcommand{\oDbf}{{\bf \breve{D}}}
\newcommand{\oAbf}{{\bf \breve{A}}}
\renewcommand \vec{{\mbox{vec}}}
\newcommand{\Acalbf}{\bf {\cal A}}
\newcommand{\calZbf}{\mbox{\boldmath $\cal Z$}}
\newcommand{\feop}{\hfill \rule{2mm}{2mm} \\}
\newtheorem{theorem}{Theorem}[section]

%-------My Definition-------
\newcommand{\Rnum}{{\mathbb R}}
\newcommand{\Cnum}{{\mathbb C}}
\newcommand{\Znum}{{\mathbb Z}}
\newcommand{\Enum}{{\mathbb E}}

\newcommand{\Qcal}{{\cal Q}}
\newcommand{\Mcal}{{\cal M}}
\newcommand{\Ccal}{{\cal C}}
\newcommand{\Dcal}{{\cal D}}
\newcommand{\Hcal}{{\cal H}}
\newcommand{\Ocal}{{\cal O}}
\newcommand{\Rcal}{{\cal R}}
\newcommand{\Zcal}{{\cal Z}}
\newcommand{\Xcal}{{\cal X}}
\newcommand{\zzbf}{{\bf 0}}
\newcommand{\zebf}{{\bf 0}}

\newcommand{\eop}{\hfill $\Box$}

%---
\newcommand{\gss}{\mathop{}\limits}
\newcommand{\gs}{\mathop{\gss_<^>}\limits}
%---usage $$\gs_{H_1}^{H_0}$

\newcommand{\circlambda}{\mbox{$\Lambda$
             \kern-.85em\raise1.5ex
             \hbox{$\scriptstyle{\circ}$}}\,}

\newcommand{\tr}{\mathop{\rm tr}}
\newcommand{\var}{\mathop{\rm var}}
\newcommand{\cov}{\mathop{\rm cov}}
\newcommand{\diag}{\mathop{\rm diag}}
\def\rank{\mathop{\rm rank}\nolimits}
\newcommand{\ra}{\rightarrow}
\newcommand{\ul}{\underline}
\def\Pr{\mathop{\rm Pr}}
\def\Re{\mathop{\rm Re}}
\def\Im{\mathop{\rm Im}}

\def\submbox#1{_{\mbox{\footnotesize #1}}}
\def\supmbox#1{^{\mbox{\footnotesize #1}}}

%%%%%%%%   ``Theorem-like'' environments (defs, lemmas numbered like Theorems)
%
\newtheorem{Theorem}{Theorem}[section]
\newtheorem{Definition}[Theorem]{Definition}
\newtheorem{Proposition}[Theorem]{Proposition}
\newtheorem{Lemma}[Theorem]{Lemma}
\newtheorem{Corollary}[Theorem]{Corollary}
%
% to label, reference them
%
\newcommand{\ThmRef}[1]{\ref{thm:#1}}
\newcommand{\ThmLabel}[1]{\label{thm:#1}}
\newcommand{\DefRef}[1]{\ref{def:#1}}
\newcommand{\DefLabel}[1]{\label{def:#1}}
\newcommand{\PropRef}[1]{\ref{prop:#1}}
\newcommand{\PropLabel}[1]{\label{prop:#1}}
\newcommand{\LemRef}[1]{\ref{lem:#1}}
\newcommand{\LemLabel}[1]{\label{lem:#1}}
%
%%%%%%%%

%\newcommand \abs{{\boldsymbol a}}
\newcommand \bbs{{\boldsymbol b}}
\newcommand \cbs{{\boldsymbol c}}
\newcommand \dbs{{\boldsymbol d}}
\newcommand \ebs{{\boldsymbol e}}
\newcommand \fbs{{\boldsymbol f}}
\newcommand \gbs{{\boldsymbol g}}
\newcommand \hbs{{\boldsymbol h}}
\newcommand \ibs{{\boldsymbol i}}
\newcommand \jbs{{\boldsymbol j}}
\newcommand \kbs{{\boldsymbol k}}
\newcommand \lbs{{\boldsymbol l}}
\newcommand \mbs{{\boldsymbol m}}
\newcommand \nbs{{\boldsymbol n}}
\newcommand \obs{{\boldsymbol o}}
\newcommand \pbs{{\boldsymbol p}}
\newcommand \qbs{{\boldsymbol q}}
\newcommand \rbs{{\boldsymbol r}}
\newcommand \sbs{{\boldsymbol s}}
\newcommand \tbs{{\boldsymbol t}}
\newcommand \ubs{{\boldsymbol u}}
\newcommand \vbs{{\boldsymbol v}}
\newcommand \wbs{{\boldsymbol w}}
\newcommand \xbs{{\boldsymbol x}}
\newcommand \ybs{{\boldsymbol y}}
\newcommand \zbs{{\boldsymbol z}}

\newcommand \Bbs{{\boldsymbol B}}
\newcommand \Cbs{{\boldsymbol C}}
\newcommand \Dbs{{\boldsymbol D}}
\newcommand \Ebs{{\boldsymbol E}}
\newcommand \Fbs{{\boldsymbol F}}
\newcommand \Gbs{{\boldsymbol G}}
\newcommand \Hbs{{\boldsymbol H}}
\newcommand \Ibs{{\boldsymbol I}}
\newcommand \Jbs{{\boldsymbol J}}
\newcommand \Kbs{{\boldsymbol K}}
\newcommand \Lbs{{\boldsymbol L}}
\newcommand \Mbs{{\boldsymbol M}}
\newcommand \Nbs{{\boldsymbol N}}
\newcommand \Obs{{\boldsymbol O}}
\newcommand \Pbs{{\boldsymbol P}}
\newcommand \Qbs{{\boldsymbol Q}}
\newcommand \Rbs{{\boldsymbol R}}
\newcommand \Sbs{{\boldsymbol S}}
\newcommand \Tbs{{\boldsymbol T}}
\newcommand \Ubs{{\boldsymbol U}}
\newcommand \Vbs{{\boldsymbol V}}
\newcommand \Wbs{{\boldsymbol W}}
\newcommand \Xbs{{\boldsymbol X}}
\newcommand \Ybs{{\boldsymbol Y}}
\newcommand \Zbs{{\boldsymbol Z}}

\newcommand \Absolute[1]{\left\lvert #1 \right\rvert}

\title{RIS-Assisted Self-Interference Mitigation for In-Band Full-Duplex Transceivers}

\author{Wei Zhang, Yi Jiang,  \textit{Member, IEEE}, and Bin Zhou  % \\  %   % <-this % stops a space
%Key Laboratory for Information Science of Electromagnetic Waves (MoE) \\
%Shanghai Institute for Advanced Communication and Data Science \\
%Emails:17210720089@fudan.edu.cn, yijiang@fudan.edu.cn, xwang11@fudan.edu.cn
}


\maketitle
\blfootnote{
% Work in this paper was supported by National Natural Science Foundation of China Grant No. 61771005. {(\it Corresponding author: Yi Jiang.)}

W. Zhang and Y. Jiang are with the Key Laboratory for Information
Science of Electromagnetic Waves (MoE), Department of Communication
Science and Engineering, School of Information Science and Technologies,
Fudan University, Shanghai 200433, China (emails: wzhang19@fudan.edu.cn, yijiang@fudan.edu.cn).

B. Zhou is with the Key Laboratory of Wireless Sensor Network and Communications, Chinese Academy of Sciences (CAS) (e-mail: bin.zhou@mail.sim.ac.cn).
}
\begin{abstract}
The wireless in-band full-duplex (IBFD) technology can in theory double the system capacity over the conventional frequency division duplex (FDD) or time-division duplex (TDD) alternatives. But the strong self-interference of the IBFD can cause excessive quantization noise in the analog-to-digital converters (ADC), which represents the hurdle for its real implementation. In this paper, we consider employing a reconfigurable intelligent surface (RIS) for IBFD communications. While the BS transmits and receives the signals to and from the users simultaneously on the same frequency band, it can adjust the reflection coefficients of the RIS to configure the wireless channel so that
the self-interference of the BS is sufficiently mitigated in the propagation domain. Taking the impact of the quantization noise into account, we propose to jointly design the downlink (DL) precoding matrix and the RIS coefficients to  maximize the sum of uplink (UL) and DL rates. The effectiveness of the proposed RIS-assisted in-band full-duplex (RAIBFD) system is verified by simulation studies, even taking into considerations that the phases of the RIS have only finite resolution. % that show that its sum-rate outperforms that of the conventional half-duplex and the other existing methods.

\end{abstract}
%In practice, however, .

% \fi
\IEEEpeerreviewmaketitle

\section{Introduction}
Owing to the ever-increasing demand on the capacity of next-generation wireless systems, researchers have kept questing techniques that can achieve higher spectral efficiency, among which the in-band full-duplex (IBFD) wireless have deservedly attracted considerable attentions \cite{6832464,7169508}. In contrast to the conventional time division duplex (TDD) and frequency division duplex (FDD) mode, an IBFD node can simultaneously transmit and receive signals (STAR) to double the system capacity, at least in theory. But for real implementation, the super strong self-interference due to the STAR \cite{6736751} can saturate the analog-digital converters (ADCs) of the adjacent receivers. The resultant excessively large quantization noise can void the reception of desired signal; thus, it is crucial to mitigate the self-interference at the receiver antennas before the ADCs \cite{5985554,6280258,7105651}.
The mitigation of strong self-interference has gained much attention in recent years.
The existing mitigation methods can be divided into two major categories based on whether they utilize the known waveform of self-interference or not.

Among the methods using the waveform of self-interference, researches have been devoted to cancel the self-interference using an auxiliary radio frequency (RF) chain \cite{5757799,1107.0607,6190376,6353396,6517516,6656015,9153162}. In general, this approach consists of two steps: first, both the channel responses of wireless self-interference channel and the auxiliary RF chain are estimated during the pilot training period; second, a reconstructed version of the self-interference can be created in the digital domain based on the estimated channels, and transmitted by the auxiliary RF chain to the receiver RF chains, which is subtracted from the received signal before the ADCs.
Another approach proposed was to utilize an analog canceller to transform the output signal of power amplifier (PA) into a copy of self-interference, and subtract it from the received signal before the ADCs at receiver \cite{10.1147,10.1145,7913706,6648617,7426862}.
The method proposed in \cite{10.1147} utilizes an analog transformer called balun to obtain an inverted copy of the transmit signal, of which the delay and attenuation is further adjusted by a circuit to match and cancel the self-interference at the receiver antenna. As a near-perfect match of the delay is extremely difficult \cite{10.1147}, another analog canceller is introduced in \cite{10.1145}, which uses the linear combination of several fixed delay versions of the copied self-interference to approximate the actual self-interference.  The paper \cite{7913706} uses the Taylor series to approximate the delayed self-interference with the original one and its derivative, which reduces the number of unknown parameters and facilitates the implementation of the circuit in analog domain. The paper \cite{6648617} applied the steepest-descent method to the cost function defined as the power of the analog canceller output, and obtained the weights of analog canceller until the cost function converges. Taking the nonlinearity of the PA, the non-ideal nature of the attenuator and phase shifter into account, the paper \cite{7426862} proposed an adaptive dithered linear search (DLS) method to minimize the receive signal strength indicator (RSSI) by optimizing the amplitudes and phases of different taps in the tapped delay line architecture.

The other category of researches  proposed to mitigate the strong interference without assuming that the waveform of interference is known. The main advantage of this category of works lies in the fact that the aforementioned works assuming known waveform of the interference are susceptible to the nonideality of the RF hardware -- imperfect reconstruction of  the self-interference in the analog domain can significantly degenerate the system performance.
The paper \cite{101145} proposed to place two transmit antennas at a distance $d$ and $d+\frac{\lambda}{2}$ away from the receive antenna for the signal destructive superposition. However, this method can only support a single stream data transmission. As an improvement to \cite{101145}, the work in \cite{10.1148} presented the first MIMO full duplex wireless system called MIDU. By utilizing the symmetric placement of transmit and receive antennas, the MIDU technology can both mitigate the self-interference and transmit multiple streams of data. Given the perfect channel state information (CSI), the paper \cite{SoftNull} proposed a so-called SoftNull method which utilizes a transmit beamforming matrix to mitigate the strong self-interference before it enters into the ADCs.
The paper \cite{9420257} proposed a technique named HIMAP to mitigate the strong interference using an analog phase shifter network (PSN), where both the waveform and the channel state information (CSI) are not needed. But the performance of HIMAP is sensitive to the phase deviations of the PSN, for which the authors also provided an over-the-air (OTA) calibration method to eliminate the effect of phase deviations in \cite{9810495}. It is worthwhile noting that the two categories of the interference mitigation methods are not exclusive to each other. They can potentially be combined to achieve even higher performance.

In recent years ,the advent of reconfigurable intelligent surface (RIS) technology offers a new option for self-interference mitigation (SIM), as it can reconfigure the wireless channels by dynamically adjusting the reflection elements  \cite{cuiRIS}\cite{8910627}.  In addition, the RIS as a passive device has benefit of low power consumption and introducing no circuit noise.
In this paper, we investigate a RIS-assisted in-band full-duplex (RAIBFD) system, where a  base-station (BS), assisted by a RIS, receives signal from the UL users and transmit signal to the DL users simultaneously.
As the self-interference reaches the receive antennas as a superposition of the direct path from the transmitter to the receiver and the reflection path from the RIS,
we can adjust the RIS so that the self-interference can add destructively at the receive antennas of the BS to avoids the saturation of ADCs. We aim at mitigating the strong self-interference by optimizing the transmit beamforming matrix of the BS and the reflection coefficients of the RIS, and further obtain the optimal power allocation for DL users.
The proposed RIS-assisted SIM method requires no analog cancellers, and has good compatibility with existing communication systems.
The simulation results verify the RAIBFD can achieve sum-rate $36\%$ higher the state-of-the-art methods SoftNull \cite{SoftNull} and $48\%$ over the FDD counterpart that is also assisted by a RIS \cite{RISFDD}. The proposed RAIBFD system is inclusive; it can potentially also employ the aforementioned methods, e.g., \cite{5757799,1107.0607,6190376,6353396,6517516,6656015,9153162}, to achieve even higher performance.

This paper is organized as follows. Section \ref{SEC2} introduces the RIS-assisted in-band full-duplex (RAIBFD) signal model of both UL and DL. In Section \ref{SEC3}, a joint design scheme of the RIS coefficients and the BS precoding matrix is proposed. Section \ref{SEC4} presents an idealistic full-duplex system with $\infty$-bit ADC, and simulation results are provided in Section \ref{SEC5}. Finally, conclusions are drawn in Section \ref{SEC6}.

\textit{Notations:} $(\cdot)^*$, $(\cdot)^T$, and $(\cdot)^H$ denotes the conjugate, the transpose, and conjugate transpose of a matrix, respectively; $||\cdot||_F$ represents a matrix's Frobenius norm; $|\cdot|$ denotes the determinant of a matrix;
$*$ and $\circ$ represents the Khatri–Rao product and Hadamard product, respectively; $\abf(i)$ denotes the $i$th element of vector $\abf$; ${\rm Re}(\cdot)$ represents the real part of a matrix; $\ebf_i$ is a vector with its $i$th element being $1$ and others being $0$;
$\text{diag}(\abf)$ stands for a diagonal matrix with vector $\abf$ being the diagonal element;
$\text{diag}(\Abf)$ stands for column vectors with its elements being those diagonal elements of matrix $\Abf$;
${\rm vec}(\Abf)$ stacks the column vector of  matrix $\Abf$ into a single column; $\angle{\Abf}$ stands for a matrix of which the entries are the phases of elements in $\Abf$;
${\bf 0}_{M\times N}$ represents a $M\times N$ matrix with all elements being $0$;
$\xbf\sim {\cal CN}(\mubf,\Qbf)$ is a complex Gaussian random vector with mean $\mubf$ and covariance matrix $\Qbf$.

\section{Signal Model} \label{SEC2}
\begin{figure}[htb]
\centering
{\psfig{figure=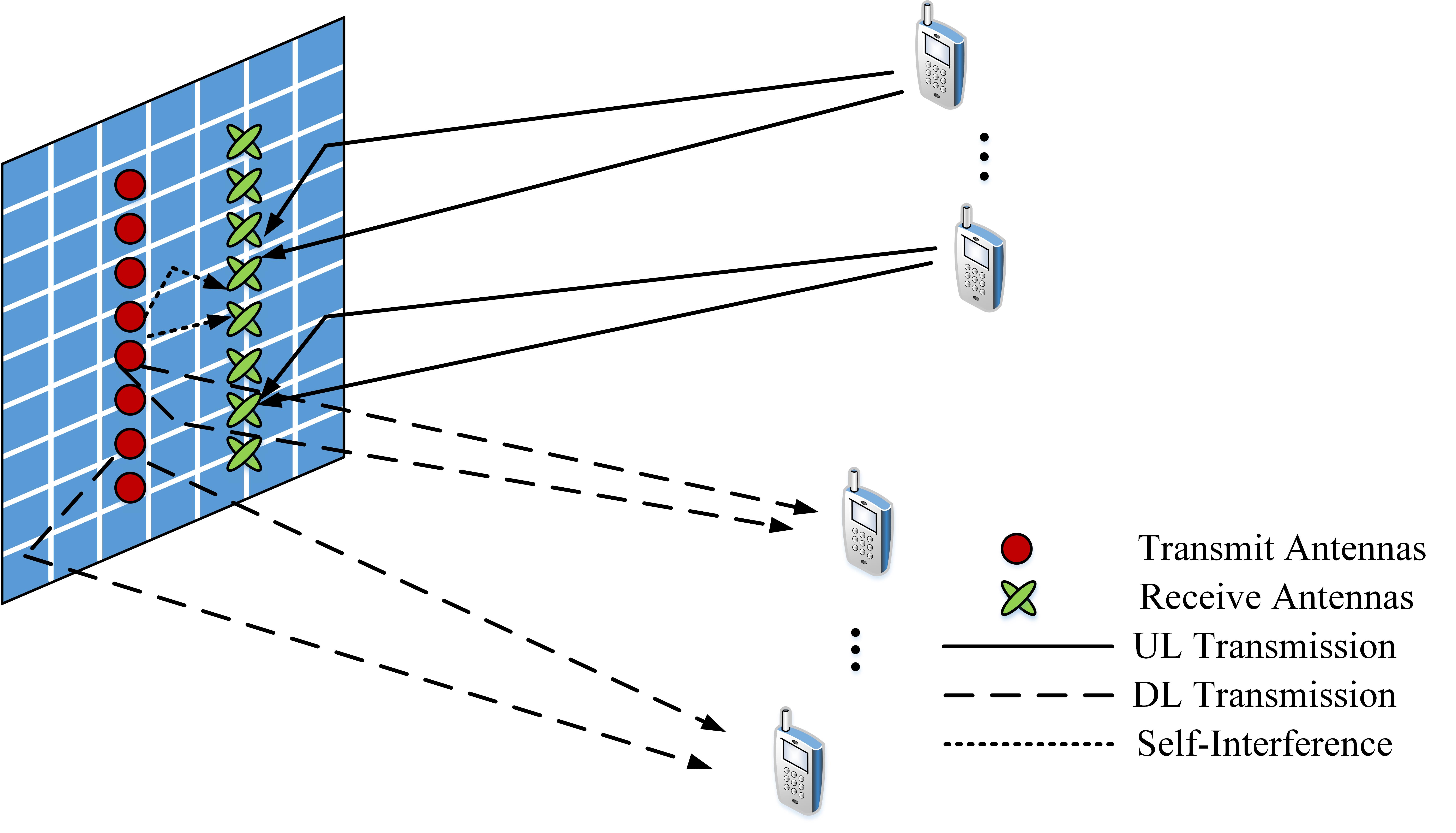,width= 3.5in}}
\caption{The framework of the RIS-based in-band full-duplex communication system.}
\label{fig.Framework}
\end{figure}
Consider a RAIBFD signal system depicted as Fig. \ref{fig.Framework} which consists of a base station (BS) equipped with $M_t$ transmitting (Tx) antennas and $M_r$ receiving (Rx) antennas, a RIS placed near to the BS with $M_{ris}$ elements, $K_u$ Tx users, and $K_d$ Rx users. Denote the channel from the Tx antennas to the Rx antennas of BS as  $\Hbf_{B_rB_t}\in{\mathbb C}^{M_r\times M_t}$, the channel from the Tx antennas of the BS to the RIS as $\Hbf_{RB_t}\in{\mathbb C}^{M_{ris}\times M_t}$, the channel from the RIS to the Rx antennas of the BS as  $\Hbf_{B_rR}\in{\mathbb C}^{M_r\times M_{ris}}$, the channel from $i$th UL user to the RIS as $\hbf_{Ru,i}\in{\mathbb C}^{M_{ris}\times 1}$, and the channel from the RIS to the $i$th DL user as $\hbf_{dR,i}^H\in{\mathbb C}^{1\times M_{ris}}$. Also denote the channel from the $i$th UL user to the BS and that from the BS to the $j$th DL user as $\hbf_{B_ru,i}\in{\mathbb C}^{M_r\times 1}$ and $\hbf_{dB_t,j}^H\in{\mathbb C}^{1\times M_{t}}$, respectively.
\subsection{Uplink Signal Model}
In the UL channel, the BS receives
\ben
\ybf_B = (\Hbf_{B_rR}\Dbf\Hbf_{Ru}+\Hbf_{B_ru})\Gammabf_u\sbf_u + \xibf_B +\zbf_B,
\label{equ.bsSig}
\een
where the first term is the signal of interest from the $K_u$ UL users. In particular,
\ben
\Dbf = {\rm diag}(e^{j\phi_1},e^{j\phi_2},\dots,e^{j\phi_{M_{ris}}}),
\een
\ben
\Hbf_{B_ru} = [\hbf_{B_ru,1},\hbf_{B_ru,2},\dots,\hbf_{B_ru,K_u}],
\een
\ben
\Gammabf_u = {\rm diag}\left(\sqrt{\beta_{u,1}},\sqrt{\beta_{u,2}},\dots,\sqrt{\beta_{u,K_u}}\right),
\een
and
\ben
\sbf_u = [s_{u,1},s_{u,2},\dots,s_{u,K_u}]^T.
\een
$\Dbf$ represents the phases of the RIS elements, and $\{\phi_i\}_{i=1}^{M_{ris}}\in\left[0,2\pi\right]$; $\{\beta_{u,k}\}_{k=1}^{K_u}$ represents the pathloss from $k$th UL user to the RIS and BS; $\{s_{u,k}\}_{k=1}^{K_u}$ are the transmitted symbols from the $k$th UL user to the BS, and $\Enum\{|s_{u,k}|^2\}=\sigma_{u}^2$; The second term of (\ref{equ.bsSig}) is the self-interference, i.e.,
\ben
\xibf_B = \Gbf_{\rm SI}\Pbf\sbf_d
\label{equ.xiB}
\een
with
\ben
\Gbf_{\rm SI} = \Hbf_{B_rR}\Dbf\Hbf_{RB_t}+\Hbf_{B_rB_t},
\label{equ.SIEff}
\een
\ben
\Pbf = [\pbf_1,\pbf_2,\dots,\pbf_{K_d}],
\een
\ben
\sbf_d = [s_{d,1},s_{d,2},\dots,s_{d,K_d}]^T.
\een
$\Gbf_{\rm SI}$ is the effective channel of self-interference; $\pbf_{k}\in{\mathbb C}^{M_{t}\times 1}$ is the precoding vector of $k$th DL user; $\{s_{d,k}\}_{k=1}^{K_d}$ are the transmitted symbols from the BS to $k$th DL user, and $\Enum\{|s_{d,k}|^2\}=\sigma_{d,k}^2$. The third term $\zbf_B\sim{\mathcal CN}({\bf 0}, \sigma_B^2\Ibf)$ is the Gaussian noise at BS.

Denoting the UL effective channel as
\ben
\Hbf_u = (\Hbf_{B_rR}\Dbf\Hbf_{Ru}+\Hbf_{B_ru})\Gammabf_u,
\label{equ.Hu}
\een
we have the output of the ADCs as \cite{7307134}
\ben
\begin{split}
\breve{\ybf}_B &= \alpha\ybf_{B} + \nbf_q, \\
&=\alpha\Hbf_u\sbf_u + \alpha\xibf_B +\alpha\zbf_B+\nbf_q,
\end{split}
\label{equ.SigMdQuantv1}
\een
where
\ben
\alpha = 1 - \rho,\ \rho = \frac{\pi\sqrt{3}}{2}2^{-2{\sf ENOB}}
\label{equ.ENOB}
\een
with ${\sf ENOB}$ being the effective number of bits of ADCs; $\nbf_q$ is additive Gaussian quantization noise vector and its covariance matrix is \cite{7307134}
\ben
\Rbf_{q} = \rho(1-\rho){\rm diag}(\Rbf_{\ybf_B}),
\label{equ.Rnq}
\een
where
\ben
\begin{split}
\Rbf_{\ybf_B} =&\Enum\{\ybf_B\ybf_B^H\}, \\
=&\sigma_{u}^2\Hbf_u\Hbf_u^H + \Gbf_{\rm SI}\Pbf\Rbf_{d}\Pbf^H\Gbf_{\rm SI}^H + \sigma_B^2\Ibf,
\end{split}
\label{equ.RyBv1}
\een
and
\ben
\Rbf_{d} = {\rm diag}\left(\sigma_{d,1}^2,\sigma_{d,2}^2,\dots,\sigma_{d,K_d}^2\right).
\label{equ.Rbfsd}
\een
Given fixed phases of the RIS elements, the self-interference $\xibf_B$ is known at the BS, and thus we can just subtract it from the received signal to obtain
\ben
\begin{split}
\tilde{\ybf}_B &= \breve{\ybf}_B - \alpha\xibf_B , \\
&=\alpha\Hbf_u\sbf_u +\alpha\zbf_B+\nbf_q.
\end{split}
\label{equ.SigMdQuantv2}
\een
We have from (\ref{equ.SigMdQuantv2}) that the UL spectral efficiency is
\ben
R_u = {\rm log}_2\left|\Ibf+\alpha^2\sigma_{s_u}^2\Hbf_u^H\Qbf_B^{-1}\Hbf_u\right|,
\label{equ.CapR1}
\een
where
\ben
\begin{split}
\Qbf_{B} &= {\mathbb E}\{(\alpha\zbf_B+\nbf_q)(\alpha\zbf_B+\nbf_q)^H\}, \\
&= \alpha^2\sigma_B^2\Ibf + \Rbf_{q}.
\end{split}
\label{equ.Qb}
\een
\subsection{Quantization Noise Caused by the Self-Interference}
The self-interference becomes a bottleneck for the UL signal transmission as the strong self-interference can saturate the ADCs and cause excessive quantization noise.
According to (\ref{equ.ENOB}) and (\ref{equ.Rnq}), the power of quantization noise is determined by  the {\sf ENOB} of ADCs and the covariance $\Rbf_{\ybf_B}$. It follows from (\ref{equ.RyBv1}) that
\ben
\Rbf_{\ybf_B} = \sigma^2_{u}\Hbf_u\Hbf_u^H  +\Hbf_{{\rm SI}}\Rbf_{d}\Hbf_{{\rm SI}}^H + \sigma_B^2\Ibf,
\label{equ.RyBv2}
\een
where $\Hbf_{{\rm SI}} = \Gbf_{\rm SI}\Pbf$.
Denoting
\ben
\Hbf_u = \left[\hbf_{{u},1},\hbf_{{u},2},\dots,\hbf_{{u},K_u}\right]
\een
and
\ben
\Hbf_{\rm SI} = \left[\hbf_{{\rm SI},1},\hbf_{{\rm SI},2},\dots,\hbf_{{\rm SI},K_d} \right],
\een
we have from (\ref{equ.RyBv2}) that
\ben
\Rbf_{\ybf_B} = \sigma_{u}^2\sum_{k=1}^{K_u}\hbf_{{u},k}\hbf_{{u},k}^H + \sum_{k=1}^{K_d}\sigma_{,k}^2\hbf_{{\rm SI},k}\hbf_{{\rm SI},k}^H+\sigma_B^2\Ibf,
\een
and we can obtain the $j$th diagonal entry of $\Rbf_{\ybf_B}$ as
\ben
\left[\Rbf_{\ybf_B}\right]_{j,j} = r_j  + q_j  + \sigma_B^2,
\een
where
\ben
r_j = \sum_{k=1}^{K_u}\sigma_{u}^2|\hbf_{u,k}(j)|^2
\een
\ben
q_j = \sum_{k=1}^{K_d}\sigma_{d,k}^2|\hbf_{{\rm SI},k}(j)|^2.
\label{equ.intf}
\een
Thus, at the output of the $j$th ADC, the quantization noise power is
\ben
\left[\Rbf_{q}\right]_{j,j} = \rho(1-\rho)(r_j+q_j+\sigma_B^2),
\een
and the signal to quantization noise ratio (SQNR) at $j$th antenna is
\ben
\begin{split}
{\sf SQNR}_j &= \frac{1}{\rho(1-\rho)}\frac{r_j}{r_j +q_j + \sigma_B^2}\\
&= \frac{1}{\rho(1-\rho)}\frac{1}{1+\frac{q_j + \sigma_B^2}{r_j}}.
\end{split}
\een
Denoting the signal to interference and noise ratio (SINR) at $j$th antenna as
\ben
{\sf SINR}_j = \frac{r_j}{q_j + \sigma_B^2},
\label{equ.SINRj}
\een
yields that
\ben
{\sf SQNR}_j =
 \frac{1}{\rho(1-\rho)}\frac{1}{1+\frac{1}{{\sf SINR}_j}}.
\een
Since ${\sf SINR}_j \ll 1$ for some strong self-interference,
we have that
\ben
{\sf SQNR}_j \approx
 \frac{1}{\rho(1-\rho)}{\sf SINR}_j,
\een
which is expressed in dB as
\ben
[{\sf SQNR}_j]_{\rm dB} \approx [{\sf SINR}_j]_{\rm dB}-10{\rm log}_{10}(\rho-\rho^2).
\label{equ.SQNRv1}
\een
Inserting $\rho = \frac{\pi\sqrt{3}}{2}\cdot2^{-2{\sf ENOB}}$ into (\ref{equ.SQNRv1}) and ignoring the higher-order small term $\rho^2$, we have
\ben
[{\sf SQNR}_j]_{\rm dB} \approx [{\sf SINR}_j]_{\rm dB}+6{\sf ENOB} -4.37.
\label{equ.sqnr}
\een

$[{\sf SINR}_j]_{\rm dB}$ can be very low owing to the strong self-interference. For instance, given $[{\sf SINR}_j]_{\rm dB} = -100$dB,  $[{\sf SQNR}]_{\rm dB} = -32$dB even for ADCs with ${\sf ENOB} = 12$bit and the UL spectral efficiency is essentially zero, despite the interference subtraction (\ref{equ.SigMdQuantv2}) in the later stage. Thus, it is crucial for the UL to mitigate the self-interference before the ADCs.

\subsection{Downlink Signal Model}
In the DL transmission, denoting the effective channel from $k$th DL user to the BS as
\ben
\gbf_{d,k} = \Hbf_{RB_t}^H\Dbf^H\hbf_{dR,k}+\hbf_{dB_t,k}
\label{equ.gdi}
\een
yields the received signal of $k$th DL user as
\ben
\begin{split}
y_{k} = \sqrt{\beta_{d,k}}\gbf_{d,k}^H\pbf_{k}s_{d,k}+\underbrace{\sqrt{\beta_{d,k}}\gbf_{d,k}^H\sum_{i\ne k}\pbf_is_{d,i} + z_{k}}_{\xi_{k}},
\label{equ.1Sig}
\end{split}
\een
$\beta_{d,k}$ is the pathloss from the BS and RIS to the $k$th DL user, and $z_{k}\sim{\cal CN}(0, \sigma^2)$ is the Gaussian noise. Then the spectral efficiency of $k$th DL user is
\ben
R_{d,k} = {\rm log}_2\left(1+\frac{\beta_{d,k}\sigma_{d,k}^2}{\sigma_{\xi_{k}}^2}|\pbf_{k}^H\gbf_{d,k}|^2\right),
\label{equ.Rd}
\een
for $k = 1,2,\dots,K_d$, where
\ben
\sigma_{\xi_{k}}^2 = \beta_{d,k}\gbf_{d,k}^H(\sum_{i\ne k}\sigma_{d,i}^2\pbf_{i}\pbf_{i}^H)\gbf_{d,k} + \sigma^2.
\een

\subsection{The Problem of UL-DL Sum-rate Maximization}
In this paper, we focus on maximizing the UL-DL sum-rate by jointly optimizing the downlink beamforming matrix $\Pbf$, the power allocation $\Rbf_d$, and the RIS phases $\Dbf$, i.e., to solve
\begin{align}
&\mathop{\max}_{\Pbf,\Rbf_d,\Dbf} R_u + \sum_{k = 1}^{K_d}R_{d,k} \nonumber\\
\ & \text{s.t.} \quad {\rm tr}\left(\Pbf\Rbf_{d}\Pbf^H\right) \le P_t,
\label{equ.objFuncOri}
\end{align}
where $R_u$ and $R_{d,k}$ are given in (\ref{equ.CapR1}) and (\ref{equ.Rd}), respectively, and $P_t$ is the maximum transmit power at BS. The problem is non-convex and the optimal solution appears intractable. We propose a SIM-based method for a near-optimal solution to (\ref{equ.objFuncOri}).

\section{Joint Design of the RIS Elements and the BS Precoding Matrices} \label{SEC3}
In this section, we propose to divide the DL precoding matrix into a precoding matrix that mitigates UL self-interference and a zero-forcing (ZF) precoding matrix that eliminates the DL inter-user interference. In the UL signal transmission, we mitigate the self-interference by jointly optimizing the RIS reflection coefficients and the self-interference precoding matrix; in the DL signal transmission we use the ZF precoding matrix and water-filling based power allocation.
\subsection{Self-Interference Mitigation}
According to (\ref{equ.intf}), we propose to mitigate the self-interference by making $q_j = 0, j = 1,2,\dots, M_r$, which is equivalent to making
\ben
\Gbf_{\rm SI}\Pbf = {\bf 0},
\label{costFuncv2}
\een
where $\Gbf_{\rm SI}$ is given in (\ref{equ.SIEff}).

We propose to divide the DL precoding matrix $\Pbf$ into
\ben
\Pbf = \Pbf_{\rm SIM}\Pbf_d,
\label{equ.Pdiv}
\een
where $\Pbf_{\rm SIM}\in{\Cnum^{M_t\times M_{d}}}$ is for SIM and is semi-unitary, and $\Pbf_d\in{\Cnum^{M_d\times K_{d}}}$ is for DL precoding. Here $M_d$ is a design parameter, which regulates the dimension of the subspace where the downlink transmitted signal can lie within. As larger $M_d$ leads to higher DL rate for more DL effective antenna and lower UL rate for less SIM capability and vice versa, the parameter $M_d$ should been chosen ($K_d \le M_d  \le  M_t$) to achieve a desirable trade-off between the UL rate and the DL rate of the RAIBFD systems, as will be shown in Section \ref{SEC5}.

%To completely mitigate the self-interference, we have
% \begin{subequations}
% \begin{align}
% &(\Hbf_{B_rR}\Dbf\Hbf_{RB_t}+\Hbf_{B_rB_t})\Pbf_{\rm SIM} = {\bf 0},\label{costFuncv3a} \\
% &\Pbf_{\rm SIM}^H\Pbf_{\rm SIM} = \Ibf. \label{costFuncv3b}
% \end{align}
% \label{costFuncv3}
% \end{subequations}
% \!\!where $\Pbf_{\rm SIM}^H\Pbf_{\rm SIM} = \Ibf$ restricts $\Pbf_{\rm SIM}$  to have $M_d$ orthonormal columns, otherwise the columns of $\Pbf_{\rm SIM}$ can be rank-deficient, and disables the DL multi-user data transmission.

We formulate the SIM problem as
\begin{align}
&\mathop{\min}_{\Pbf_{\rm SIM}, \Dbf} ||(\Hbf_{B_rR}\Dbf\Hbf_{RB_t}+\Hbf_{B_rB_t})\Pbf_{\rm SIM}||_F^2 \nonumber \\
& \text{s.t.} \qquad  \Pbf_{\rm SIM}^H\Pbf_{\rm SIM} = \Ibf,
\label{costFuncv4}
\end{align}
which can be solved via optimizing the SIM matrix $\Pbf_{\rm SIM}$ and the RIS phase matrix $\Dbf$ alternately.

\subsection{Optimization of $\Pbf_{\rm SIM}$}
Assuming the RIS reflection phase matrix is fixed, we need to solve
\begin{align}
&\mathop{\min}_{\Pbf_{\rm SIM}}\ ||\Gbf_{\rm SI}\Pbf_{\rm SIM}||_F^2,\nonumber \\
& \text{s.t.} \quad \Pbf_{\rm SIM}^H\Pbf_{\rm SIM} = \Ibf,
\label{costFuncv5}
\end{align}
whose solution is a matrix whose columns are $M_d$ eigenvectors corresponding to the $M_d$ smallest eigenvalues of $\Gbf_{\rm SI}^H\Gbf_{\rm SI}$.
\subsection{Optimization of $\Dbf$}
Given that $\Pbf_{\rm SIM}$ is known, we can reformulate (\ref{costFuncv4}) into
\begin{align}
&\mathop{\min}_{\Dbf}\ ||\Hbf_{B_rR}\Dbf\Abf+\Bbf||_F^2,\nonumber \\
& \text{s.t.} \quad
\Dbf = {\rm diag}\left(e^{j\phi_1},e^{j\phi_2},\dots,e^{j\phi_{M_{ris}}}\right),
\label{costFuncv8}
\end{align}
where $\Abf = \Hbf_{RB_t}\Pbf_{\rm SIM}$ and $\Bbf = \Hbf_{B_rB_t}\Pbf_{\rm SIM}$. Using the formula ${\rm vec}(\Xbf\Ybf\Zbf) = (\Zbf^T*\Xbf)\ybf$ where $\Ybf = {\rm diag}(\ybf)$ is a diagonal matrix, we have from (\ref{costFuncv8}) that
\begin{align}
&\mathop{\min}_{\dbf}\ ||\Cbf\dbf+\bbf||_2^2,\nonumber \\
& \text{s.t.} \quad
\dbf = \left[e^{j\phi_1},e^{j\phi_2},\dots,e^{j\phi_{M_{ris}}}\right]^T,
\label{costFuncv9}
\end{align}
where
\ben
\Cbf = \Abf^T*\Hbf_{B_rR}, \ \bbf = {\rm vec}(\Bbf).
\label{equ.Cadb}
\een
Here, we propose a manifold optimization to solve (\ref{costFuncv9}) as follows.

According to the theory of manifold optimization \cite{boumal2022intromanifolds}, the constant modulus constraint of (\ref{costFuncv9}) is a manifold defined as $
\Mcal_{cc}^{M_{ris}} = \left\{\dbf\in{\Cnum}^{M_{ris}}: |\dbf(1)| = |\dbf(2)| = \cdots = |\dbf({M_{ris}})| = 1\right\}$, where $\Mcal_{cc} = \left\{x\in{\Cnum}:|x| = 1\right\}$ is a complex circle. The search space of optimization problem (\ref{costFuncv9}) consists of $M_{ris}$ complex circles, which forms a Riemannian submanifold. Thus, we can obtain a near-opimal solution to  (\ref{costFuncv9}) with a Riemannian conjugate gradient (RCG) descent algorithm over the Riemannian submanifold.
Defining
\ben
f(\dbf) \triangleq ||\Cbf\dbf+\bbf||_2^2,
\label{equ.fd}
\een
we have the Euclidean gradient $\nabla f(\dbf)$  as
\ben
\nabla f(\dbf) = \Cbf^H(\Cbf\dbf+\bbf).
\label{equ.EucGrad}
\een
The tangent space at point $\dbf$ is
\ben
T_{\dbf}\Mcal_{cc}^{M_{ris}} = \left\{\xbf\in\Cnum^{M_{ris}}: {\rm Re}\{\xbf \circ \dbf^*\} = {\bf 0}_{M_{ris}} \right\},
\een
and the Riemannian gradient at $\dbf$ is the orthogonal projection of the Euclidean gradient $\nabla f(\dbf)$ onto the tangent space $T_{\dbf}\Mcal_{cc}^{M_{ris}}$, which is
\ben
\begin{split}
\gbf &= {\rm grad}(\dbf), \\
&= \nabla f(\dbf) -{\rm Re}\{\nabla f(\dbf)\circ \dbf^*\}\circ \dbf.
\end{split}
\label{equ.RiemannGrad}
\een
Thus, given a point $\dbf_{i}, i \ge 0$ with Riemannian gradient $\gbf_{i} = {\rm grad}(\dbf_{i})$, we can obtain a conjugate direction as
\ben
\cbf_{i} = \left\{
\ba{ll} -\gbf_{0}, & i = 0. \\
-\gbf_{i} + \beta_{i}\cbf_{i-1}^{+}, & i \ge 1.
\ea \right.
\label{equ.ci}
\een
For $i \ge 1$, $\beta_{i}$ is Fletcher-Reeves parameter defined as
\ben
\beta_{i}\triangleq \frac{||\gbf_{i}||_2^2}{||\gbf^+_{i-1}||_2^2},i\ge 1.
\label{equ.PR}
\een
$\cbf_{i-1}^{+}$ and $\gbf_{i-1}^{+}$ are the transport of tangent vector $\cbf_{i-1}$ and $\gbf_{i-1}$ from $\dbf_{i-1}$ to $\dbf_{i}$, which are specified as
\ben
\begin{split}
\cbf_{i-1}^{+} &= \cbf_{i-1} - {\rm Re}\{\cbf_{i-1}\circ \dbf_{i}^*\}\circ\dbf_{i}, i\ge 1,\\
\gbf_{i-1}^{+} &= \gbf_{i-1} - {\rm Re}\{\gbf_{i-1}\circ \dbf_{i}^*\}\circ\dbf_{i}, i\ge 1.
\end{split}
\label{equ.trasp}
\een

Given the conjugate direction $\cbf_i$ from (\ref{equ.ci}), we can obtain
\ben
\tbf_{i} = \dbf_{i}+\alpha_{i}\cbf_{i}, i \ge 0
\label{equ.t}
\een
in the tangent space of point $\dbf_{i}$, where $\alpha_{i}$ is step size obtained by the Armijo{-}Goldstein condition. But we need to use retraction to map $\tbf_{i}$ onto the manifold,  which is specified as
\ben
\begin{split}
{\rm Retr}_{\dbf_{i}}:& T_{\dbf_{i}}{\Mcal_{cc}^{M_{ris}}} \rightarrow {\Mcal_{cc}^{M_{ris}}}: \\
&\alpha_{i}\cbf_{i} \rightarrow {\rm Retr}_{\dbf_{i}}(\alpha_{i}\cbf_{i}), i\ge 0,
\end{split}
\een
where
\ben
{\rm Retr}_{\dbf_{i}}(\alpha_{i}\cbf_{i}) =
\left[\frac{\tbf_{i}(1)}{|\tbf_{i}(1)|},\frac{\tbf_{i}(2)}{|\tbf_{i}(2)|},\dots,\frac{\tbf_{i}(M_{ris})}{|\tbf_{i}(M_{ris})|}\right]^T.
\label{equ.retraction}
\een
We can obtain $\dbf_{i+1}$ from $\dbf_{i}$ as
\ben
\dbf_{i+1} = {\rm Retr}_{\dbf_{i}}(\alpha_{i}\cbf_{i}).
\label{equ.xi}
\een
Iterating $i$ from $0$ to $\infty$ until (\ref{equ.fd}) convergence, we can obtain a near-optimal solution $\dbf$ over the Riemannian manifold.
The Riemannian conjugate gradient (RCG) algorithm is summarized in Algorithm \ref{Algo.1}.
\begin{algorithm}[ht]
\caption{RCG Algorithm for RIS Coefficients Optimization}
\label{Algo.1}
\begin{algorithmic}[1]
\REQUIRE  $\Cbf$, $\bbf$, and $\Dbf_0$ with randomly initialization;
\ENSURE The RIS matrix $\Dbf$;
\STATE $\dbf_{0} = \diag(\Dbf_0)$, $\cbf_0 = -{\rm grad}(\dbf_0)$, $i = 0$;
\WHILE {the cost function in (\ref{costFuncv9}) still decreases}
\STATE Obtain the step size $\alpha_{i}$ in (\ref{equ.t}) with the Armijo{-}Goldstein condition.
\STATE Obtain $\tbf_{i}$ from (\ref{equ.t}), use retraction in (\ref{equ.retraction}) to obtain next point $\dbf_{i+1}$.
\STATE Calculate the Riemannian gradient $\gbf_{i+1}={\rm grad}f(\dbf_{i+1})$ from (\ref{equ.EucGrad}) and (\ref{equ.RiemannGrad}).
\STATE Obtain the transports vector $\gbf_i^{+}$ and $\cbf_i^{+}$ of Riemannian gradient $\gbf_{i}$ and the conjugate direction $\cbf_{i}$ from $\dbf_{i}$ to $\dbf_{i+1}$ from (\ref{equ.trasp}).
\STATE Calculate Fletcher-Reeves parameter $\beta_{i+1}$ from (\ref{equ.PR}).
\STATE Obtain the conjugate direction $\cbf_{i+1}$ from (\ref{equ.ci}).
\STATE $i = i+1$.
\ENDWHILE
\end{algorithmic}
\end{algorithm}

Based on the alternating optimization (AO) minimization of cost function (\ref{costFuncv4}) between $\Pbf_{\rm SIM}$ and $\Dbf$, we obtain the algorithm as summarized in Algorithm \ref{Algo.2}, where ${\Qcal}(\cdot)$ in line $9$ can quantize the vector element-wise to the grids $\{0,\frac{2\pi}{2^b},\frac{4\pi}{2^b},\dots,\frac{(2^b-1)2\pi}{2^b}\}$ for $b$-bit RIS coefficients.

\begin{algorithm}[ht]
\caption{AO Minimization Algorithm for SIM}
\label{Algo.2}
\begin{algorithmic}[1]
\REQUIRE channel matrices $\Hbf_{B_rR}$, $\Hbf_{RB_t}$, and $\Hbf_{B_rB_t}$;
\ENSURE The SIM matrix $\Pbf_{\rm SIM}$ and the RIS matrix $\Dbf$;
\STATE Randomly initialize the RIS response matrix $\Dbf$.
\WHILE{the cost function in (\ref{costFuncv4}) larger than $10^{-10}$}
\STATE Obtain the effective channel $\Gbf_{\rm SI}$, update $\Pbf_{\rm SIM}$ with $M_d$ eigenvectors corresponding to the $M_d$ smallest eigenvalues of $\Gbf_{\rm SI}^H\Gbf_{\rm SI}$.
\STATE Fix $\Pbf_{\rm SIM}$, and calculate $\Cbf$, $\bbf$ from (\ref{equ.Cadb}).
\STATE Use Algorithm \ref{Algo.1} to obtain $\dbf$.
\IF{$b =\infty$}
\STATE$\Dbf = \diag(\dbf)$.
\ELSE
\STATE $\dbf_q = e^{j{{\Qcal}(\angle{\dbf})}}$.
\IF{$||\Cbf\dbf_q+\bbf||_2^2<||\Cbf\dbf+\bbf||_2^2$}
\STATE $\Dbf = \diag(\dbf_q)$.
\ENDIF
\ENDIF
\ENDWHILE
\end{algorithmic}
\end{algorithm}
% \subsubsection{The Number of DL Effective Antennas}
% The number of unknown variables and constraints of (\ref{costFuncv3a}) are $2M_tM_d+M_{ris}$ and $2M_rM_d$, and (\ref{costFuncv3b}) provides extra $(M_d+1)M_d$  constraints. Thus, we have
% \ben
% 2M_tM_d+M_{ris} \ge 2M_rM_d + (M_d+1)M_d,
% \een
% which yields that
% \ben
% M_d-\frac{M_{ris}}{M_d}\le2(M_t-M_r)-1.
% \label{equ.Md}
% \een
% For $M_t=M_r$, we have from (\ref{equ.Md}) that
% \ben
% M_d^2 + M_d - M_{ris} \le 0,
% \een
% and further obtain
% \ben
% M_d \le \left\lfloor\sqrt{M_{ris} - \frac{3}{4}} -\frac{1}{2}\right\rfloor.
% \een
% Owing to $M_d \in [K_d, M_t]$, thus we have
% \ben
% M_d \in [K_d, {\rm min}(M_t, \left\lfloor\sqrt{M_{ris} - \frac{3}{4}} -\frac{1}{2}\right\rfloor)],
% \een
% from which we can learn that more elements of RIS leads to more DL effective antennas.

\subsection{The Downlink Precoding}
After $\Pbf_{\rm SIM}$ and $\Dbf$ is obtained to mitigate the self-interference,  the quantization noise caused by the self-interference can be ignored. Thus, it follows from (\ref{equ.CapR1}) that the UL spectral efficiency is
\ben
R_u = {\rm log}_2\left|\Ibf+\frac{\sigma_{s_u}^2}{\sigma_B^2}\Hbf_u\Hbf_u^H\right|,
\een
which is independent of the DL precoding matrix $\Pbf_d$.  Given $\Pbf_{\rm SIM}$ and $\Dbf$, to maximize (\ref{equ.objFuncOri}) is equivalent to solving
\begin{align}
\mathop{\max}_{\Pbf_d,\Rbf_{d}} \ &\sum_{k = 1}^{K_d}R_{d,k}, \nonumber\\
\ \text{s.t.} \quad &{\rm tr}\left(\Pbf_d\Rbf_{d}\Pbf_d^H\right) \le P_t,
\label{equ.objFuncOriv2}
\end{align}
for which we consider a simple yet near-optimal ZF precoder as detailed in the following.

Denoting $\Hbf_d = \Gammabf_d(\Hbf_{dR}\Dbf\Hbf_{RB_t}+\Hbf_{dB_t})$,
where
\ben
\Hbf_{dR} = \left[\hbf_{dR,1},\hbf_{dR,2},\dots,\hbf_{dR,K_d}\right]^H,
\een
\ben
\Hbf_{dB_t} = \left[\hbf_{dB_t,1},\hbf_{dB_t,2},\dots,\hbf_{dB_t,K_d} \right]^H,
\een
and
\ben
\Gammabf_d = {\rm diag}\left(\sqrt{\beta_{d,1}},\sqrt{\beta_{d,2}},\dots,\sqrt{\beta_{d,K_d}}\right),
\een
we rewrite (\ref{equ.1Sig}) into the matrix form as
\ben
\ybf = \Hbf_d\Pbf_{\rm SIM}\Pbf_{d}\sbf_d + \zbf,
\label{equ.1Sigv2}
\een
where
\ben
\ybf = [y_1,y_2,\dots,y_{K_d}]^T,\quad \zbf = \left[z_1,z_2,\dots,z_{K_d}\right]^T.
\een
Denoting $\Hbf_{\rm SIM} \triangleq \Hbf_d\Pbf_{\rm SIM}$, we apply the ZF precoder
\ben
\Pbf_d = \Hbf_{\rm SIM}^H(\Hbf_{\rm SIM}\Hbf_{\rm SIM}^H)^{-1},
\label{equ.zfprecoding}
\een
to (\ref{equ.1Sigv2}), and we have
\ben
y_{k} = s_{d,k} + z_{k},k = 1,2,\dots, K_d,
\een
which leads (\ref{equ.objFuncOriv2}) to be
\begin{align}
&\mathop{\max}_{\Pbf_d,\Rbf_{d}} \ \sum_{k = 1}^{K_d} {\rm log}_2\left(1+\frac{\sigma_{d,k}^2}{\sigma^2}\right), \nonumber\\
& \text{s.t.} \quad {\rm tr}\left((\Hbf_{\rm SIM}\Hbf_{\rm SIM}^H)^{-1}\Rbf_{d}\right) \le P_t.
\label{equ.objFuncOriv3}
\end{align}
The solution to (\ref{equ.objFuncOriv3}) is the classic ``water filling'' power allocation.

\section{A Benchmark Study: An Idealistic Full-Duplex System with $\infty$-bit ADC} \label{SEC4}
This section investigate a RIS-assisted full-duplex system where ADCs are assumed to have $\infty$-bit resolution. By maximizing the UL-DL sum-rate of this idealistic system, we obtain a performance upper bound, which can be used to gauge the realistic system developed in Section \ref{SEC3}.

As the ${\sf ENOB} \rightarrow \infty$, we have $\rho = 0$ according to (\ref{equ.ENOB}), and then the quantization noise can be ignored in (\ref{equ.SigMdQuantv2}). Hence,
\ben
\tilde{\ybf}_B = \Hbf_u\sbf_u + \zbf_B.
\een
The UL spectral efficiency is
\ben
R_u = {\rm log}_2\left|\Ibf+\frac{\sigma_{s_u}^2}{\sigma_B^2}\Hbf_u\Hbf_u^H\right|.
\label{equ.FduLSE}
\een
For the DL signal transmission, we have from (\ref{equ.1Sigv2}) that
\ben
\ybf = \Hbf_d\Pbf\sbf_d + \zbf.
\een

Adopting the ZF precoding matrix
\ben
\Pbf = \Hbf_d^H\left(\Hbf_d\Hbf_d^H\right)^{-1},
\label{equ.Pzf}
\een
we obtain the DL spectral efficiency as
\ben
R_d = \sum_{k=1}^{K_d} {\rm log}_2\left(1+\frac{\sigma_{d,k}^2}{\sigma^2}\right).
\label{equ.HdDL}
\een
Hence  (\ref{equ.objFuncOri}) can be simplified to be
\begin{align}
&\mathop{\max}_{\{\phi_{i}\}_{i=1}^{M_{ris}},\{\sigma_{d,k}^2\}_{k=1}^{K_d}} \ R_u + R_d, \nonumber\\
& \text{s.t.} \qquad \sum_{k=1}^{K_d}\gamma_{k} \sigma_{d,k}^2\le P_t,
\label{equ.HdobjFunc}
\end{align}
where $\gamma_{k}$ is
\ben
\gamma_{k} = \left[(\Hbf_d\Hbf_d^H)^{-1}\right]_{k,k}= \tr\left(\Ebf_{k}(\Hbf_d\Hbf_d^H)^{-1}\right),
\label{equ.gammai}
\een
and $\Ebf_{k} = \ebf_{k}\ebf_{k}^T$.
When fixing the RIS phases $\{\phi_i\}_{i=1}^{M_{ris}}$, we can reduce (\ref{equ.HdobjFunc}) to
\begin{align}
&\mathop{\max}_{\{\sigma_{d,k}^2\}_{k=1}^{K_d}} \ R_d, \nonumber\\
& \text{s.t.} \quad \sum_{k=1}^{K_d}\gamma_{k} \sigma_{d,k}^2\le P_t,
\label{equ.HdobjFuncv2}
\end{align}
to which the solution is the ``water filling'' power allocation
\ben
\sigma_{d,k}^2 = {\rm max}\left(0,\frac{1}{\mu\gamma_{k}}-\sigma^2\right),k=1,2,\dots,K_d.
\label{equ.sigmad}
\een
The Lagrangian multiplier $\mu$ can be obtained from the power constraint $\sum_{k=1}^{K_d}\gamma_k \sigma_{d,k}^2 = P_t$. We assume the user selection has been done so that all the $K_d$ users are allocated with non-zero power. That is,
\ben
\mu = \frac{K_d}{\sum_{k=1}^{K_d}\gamma_{k}\sigma^2+P_t}.
\label{equ.mu}
\een
Inserting (\ref{equ.mu}) into (\ref{equ.sigmad}) yields the optimal power allocation
\ben
\sigma_{d,k}^2 = \frac{P_t+\sum_{k=1}^{K_d}\gamma_{k}\sigma^2}{K_d\gamma_{k}} - \sigma^2, k = 1,2,\dots,K_d,
\label{equ.HdWaterFillingSolu}
\een
which further leads (\ref{equ.HdDL}) to be
\ben
R_d = K_d{\rm log}_2\left(\frac{P_t}{\sigma^2}+\sum_{k=1}^{K_d}\gamma_{k}\right)
-\sum_{k=1}^{K_d}{\rm log}_2(K_d\gamma_{k}).
\label{equ.HdDLv2}
\een

After substituting  (\ref{equ.Pzf}) and  (\ref{equ.HdWaterFillingSolu}) are inserted into (\ref{equ.objFuncOri}), we recast (\ref{equ.objFuncOri}) as a function of $\dbf$:
\ben
g(\dbf) \triangleq - R_u(\dbf) - R_d(\dbf).
\een
Thus, to solve  (\ref{equ.objFuncOri}) is equivalent to
\begin{align}
\mathop{\min}_{\dbf} \ g(\dbf),
\label{equ.HdobjFuncv4}
\end{align}
of which a near-optimal solution can be obtained by the RCG algorithm in Algorithm \ref{Algo.1} with cost function $g(\dbf)$ instead of $f(\dbf)$. The Euclidean gradient of $g(\dbf)$ is given in Proposition \ref{prop.1}.

\begin{proposition}\label{prop.1}
For the cost function $g(\dbf)$, its Euclidean gradient is
\ben
\begin{split}
\nabla g(\dbf) =& -\frac{1}{{\rm ln}2}\frac{\sigma_{u}^2}{\sigma_B^2}\diag\left(\Jbf_u^H\right)
 +\frac{K_d}{{\rm ln}2}\frac{\diag\left(\Fbf_d^H\right)}{\frac{P_t}{\sigma^2}+\sum_{k=1}^{K_d}\gamma_{k}} \\ &- \frac{1}{{\rm ln}2}\sum_{k=1}^{K_d}\frac{\diag\left(\Jbf_{d,k}^H\right)}{\gamma_k},
\end{split}
\label{equ.propgderd}
\een
where
\ben
\Jbf_u = \Hbf_{Ru}\Gammabf_u\Hbf_u^H(\Ibf+\frac{\sigma_{u}^2}{\sigma_B^2}\Hbf_u\Hbf_u^H)^{-1}\Hbf_{B_rR}.
\een
\ben
\Fbf_d = \Hbf_{RB_t}\Hbf_d^H(\Hbf_d\Hbf_d^H)^{-1}(\Hbf_d\Hbf_d^H)^{-1}\Gammabf_d{\Hbf}_{dR},
\een
and
\ben
\Jbf_{d,k} = \Hbf_{RB_t}\Hbf_d^H(\Hbf_d\Hbf_d^H)^{-1}\Ebf_k(\Hbf_d\Hbf_d^H)^{-1}\Gammabf_d{\Hbf}_{dR}.
\een
\end{proposition}
\begin{proof}
The derivation of (\ref{equ.propgderd}) is detailed in Appendix.
\end{proof}

\section{Numerical Examples} \label{SEC5}
This section provides simulation results to validate the performance of the proposed RAIBFD wireless systems.
Consider a base station employed with the eight-element receive antenna array and eight-element transmit antenna array, i.e., $M_r = M_t = 8$. A RIS is placed behind the antenna arrays of the BS as illustrated in Fig. \ref{fig.BSstructure}, where (a) and (b) show the uniform linear array (ULA) and uniform rectangular array (URA), respectively.
\begin{figure}[htb]
  \begin{minipage}[htb]{0.5\linewidth}
    \centering
    \subfigure[ULA scenario]{
    \includegraphics[scale=0.25]{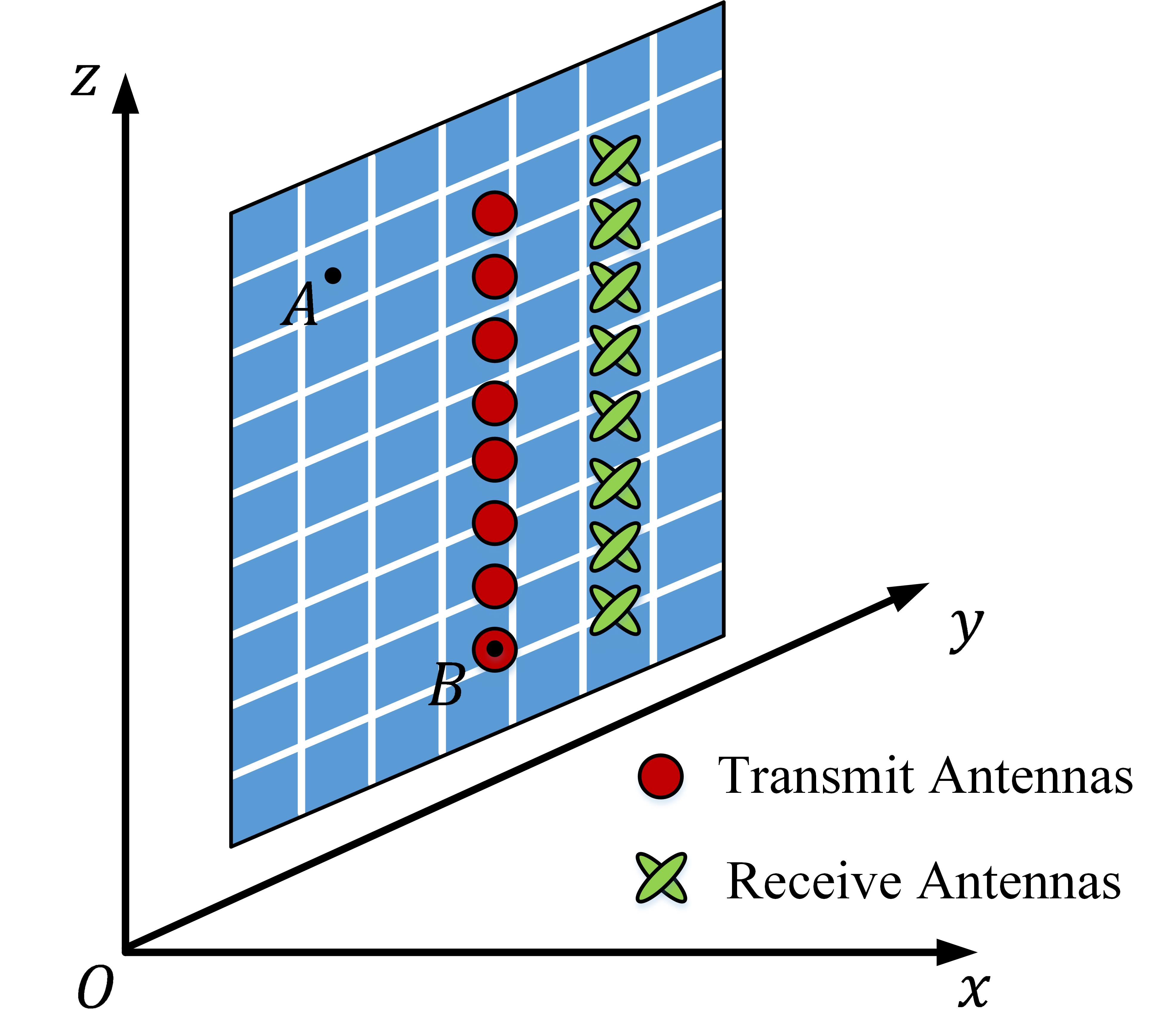}}
    \label{fig:side:a}
  \end{minipage}%
  \begin{minipage}[htb]{0.5\linewidth}
    \centering
    \subfigure[URA scenario]{
    \includegraphics[scale=0.25]{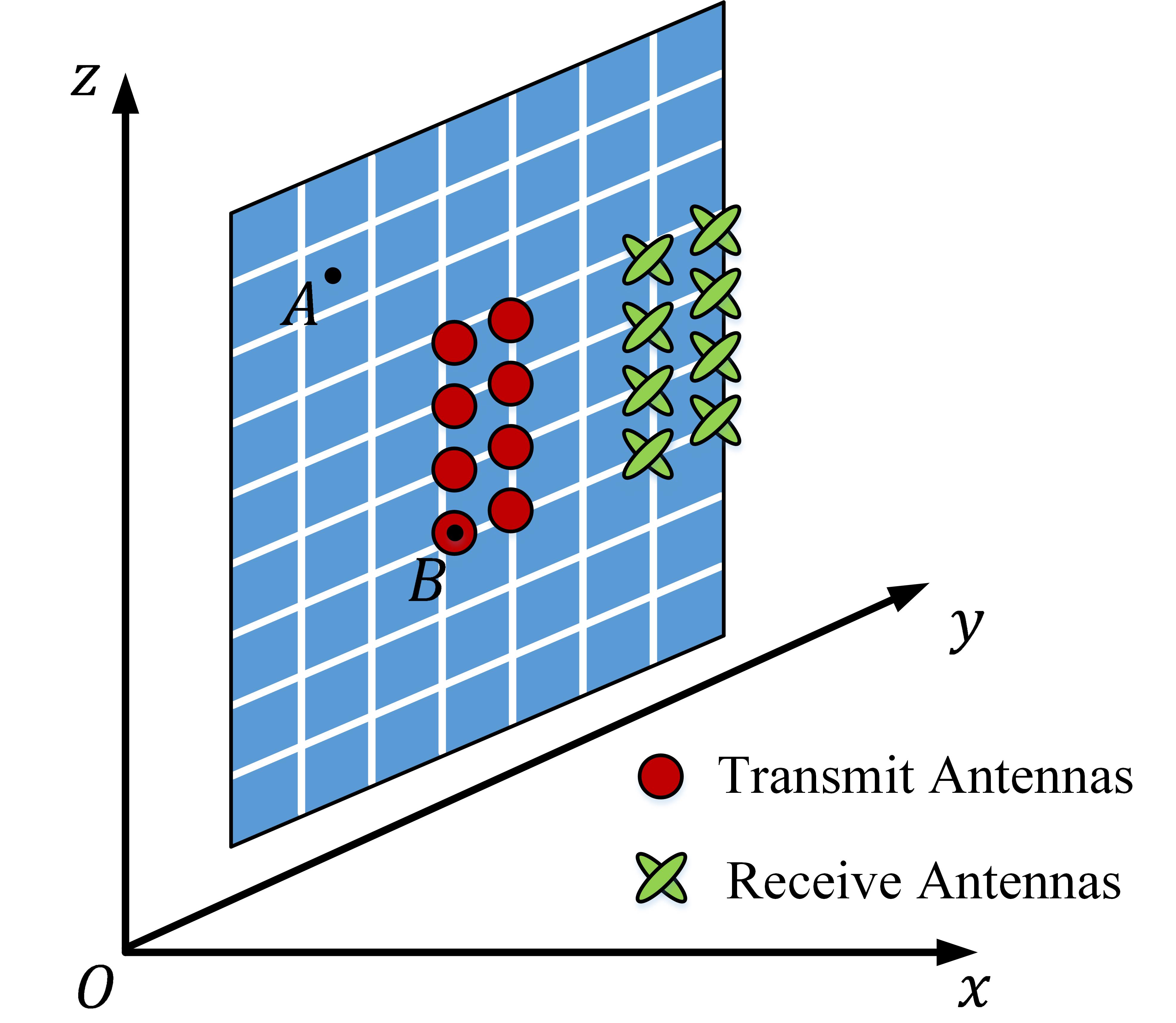}}
    \label{fig:side:b}
  \end{minipage}
  \caption{The placement of the antenna array and the RIS}
  \label{fig.BSstructure}
\end{figure}

\begin{figure*}[htb]
  \begin{minipage}[htb]{0.3\linewidth}
    \centering
    \subfigure[$4\times 4$ RIS]{
    \psfig{figure=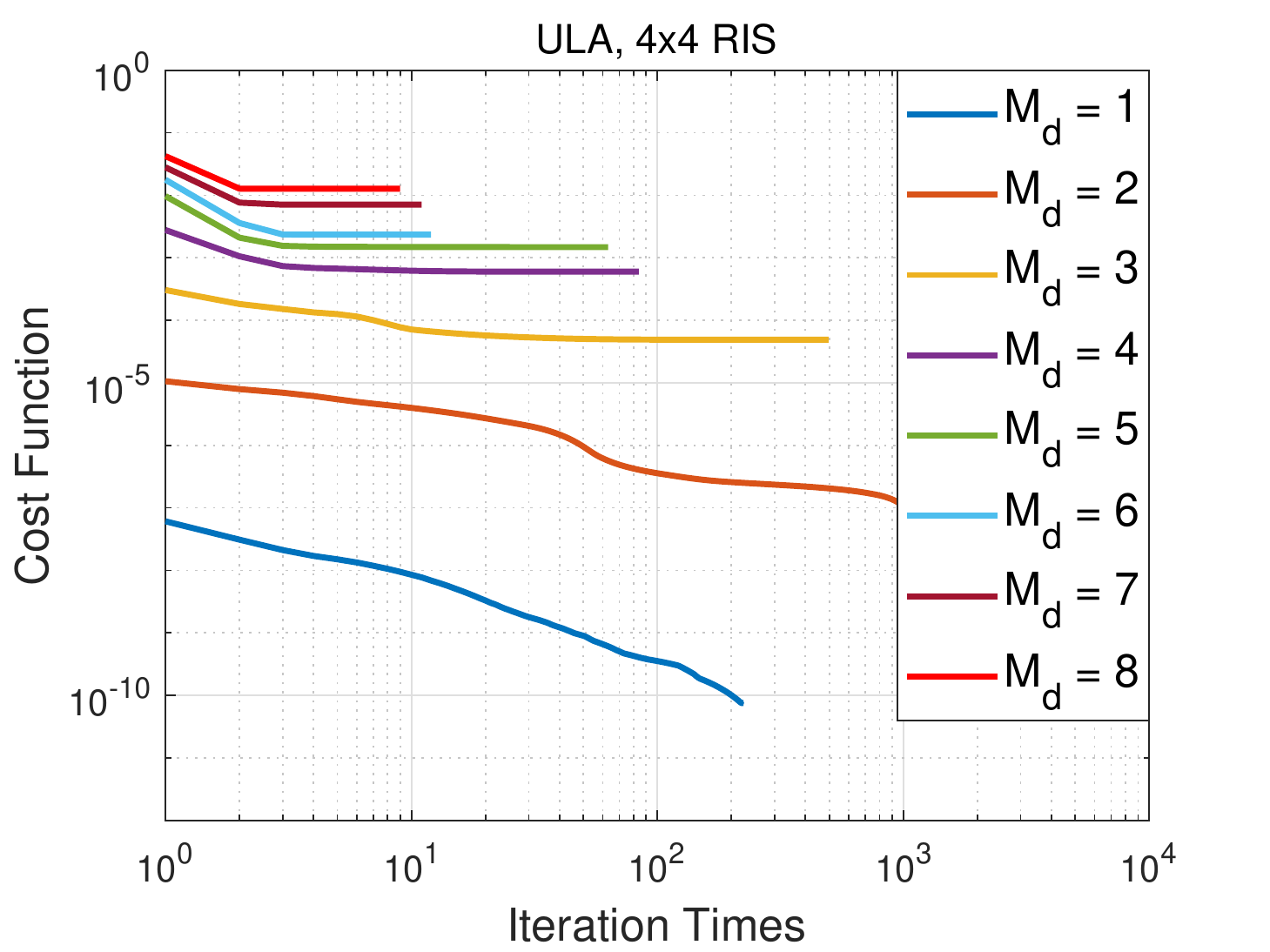,width= 2.3in}}
    \label{fig.ULAa}
  \end{minipage}%
  \begin{minipage}[htb]{0.3\linewidth}
    \centering
    \subfigure[$8\times 8$ RIS]{
    \psfig{figure=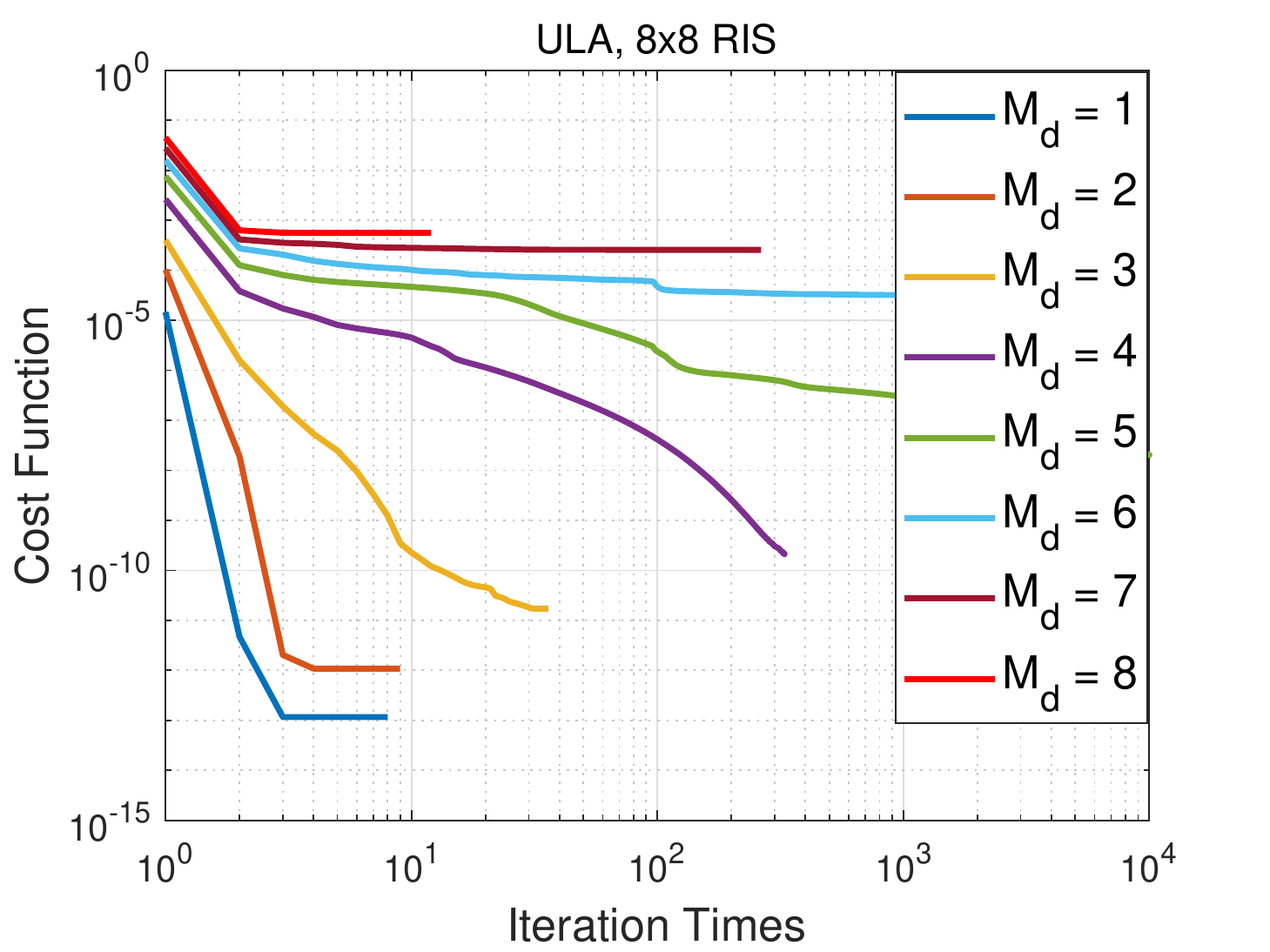,width= 2.3in}}
    \label{fig.ULAb}
  \end{minipage}
  \begin{minipage}[htb]{0.3\linewidth}
    \centering
    \subfigure[$16\times 16$ RIS]{
    \psfig{figure=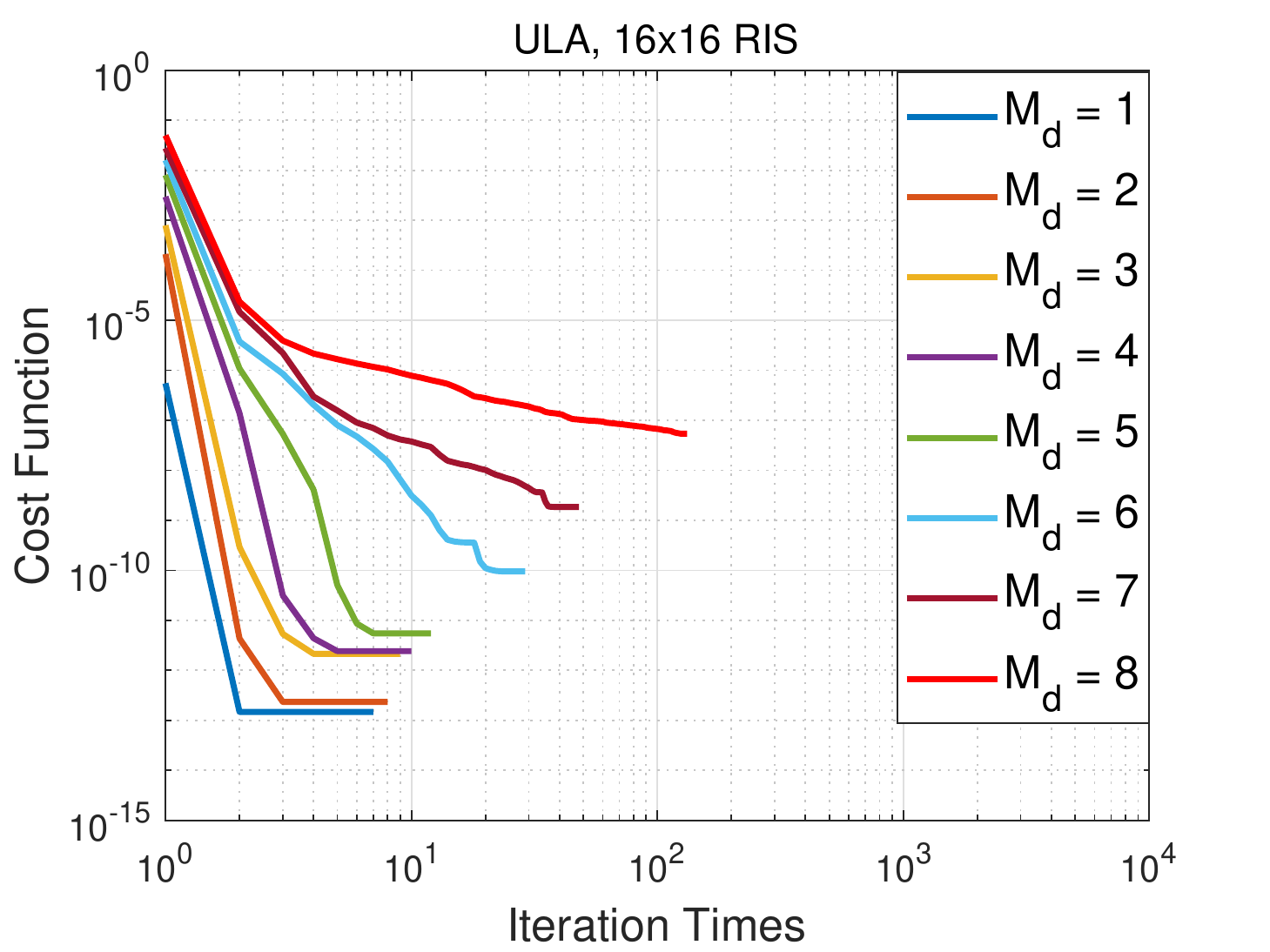,width= 2.3in}}
    \label{fig.ULAc}
  \end{minipage}
  \caption{Convergence of cost function (\ref{costFuncv4}) when running Algorithm \ref{Algo.2} for ULA shown in Fig. \ref{fig.BSstructure}(a).}
  \label{fig.ULA}
\end{figure*}
\begin{figure*}[htb]
  \begin{minipage}[htb]{0.3\linewidth}
    \centering
    \subfigure[$4\times 4$ RIS]{
    \psfig{figure=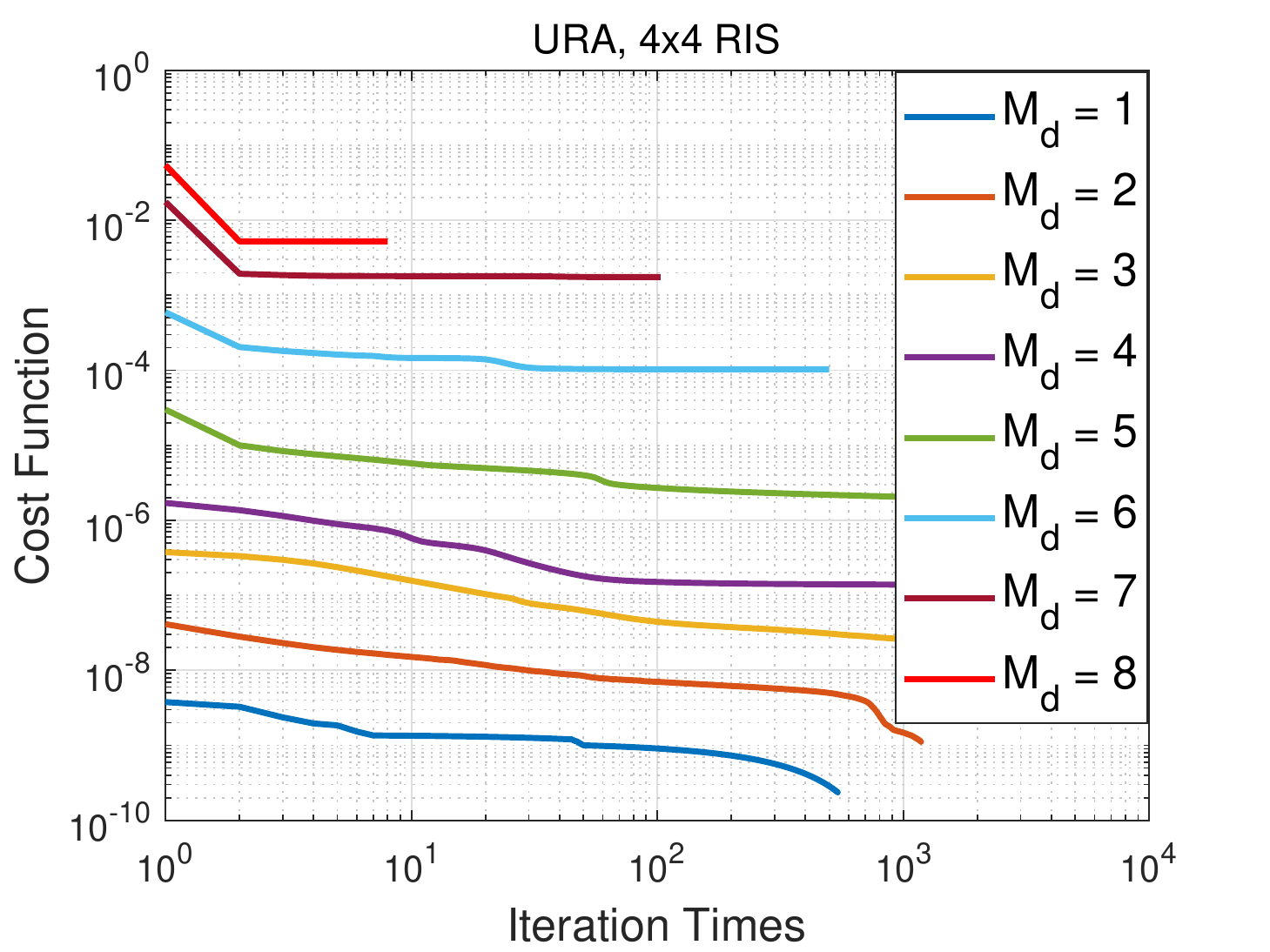,width= 2.3in}}
    \label{fig.URAa}
  \end{minipage}%
  \begin{minipage}[htb]{0.3\linewidth}
    \centering
    \subfigure[$8\times 8$ RIS]{
    \psfig{figure=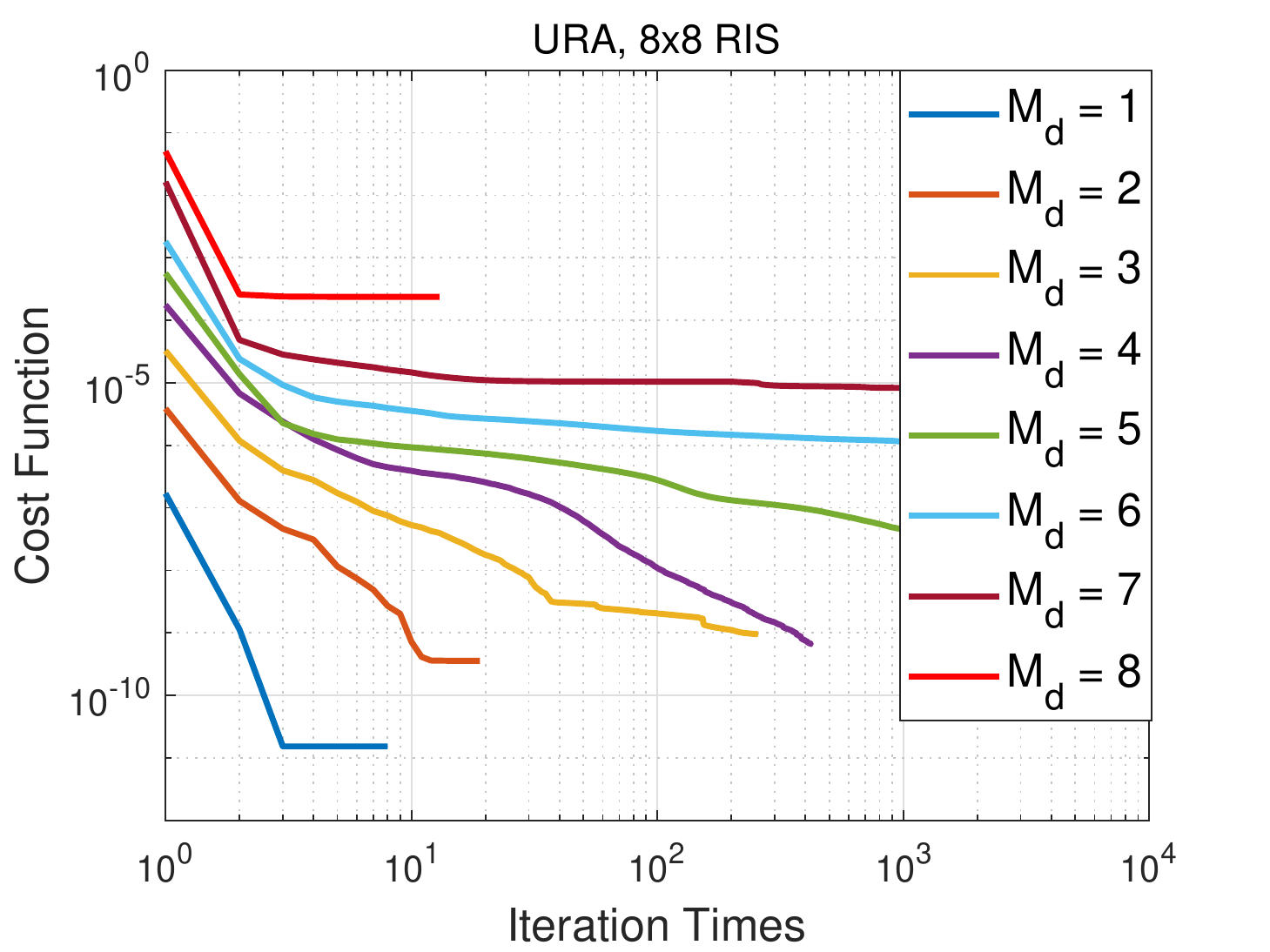,width= 2.3in}}
    \label{fig.URAb}
  \end{minipage}
  \begin{minipage}[htb]{0.3\linewidth}
    \centering
    \subfigure[$16\times 16$ RIS]{
    \psfig{figure=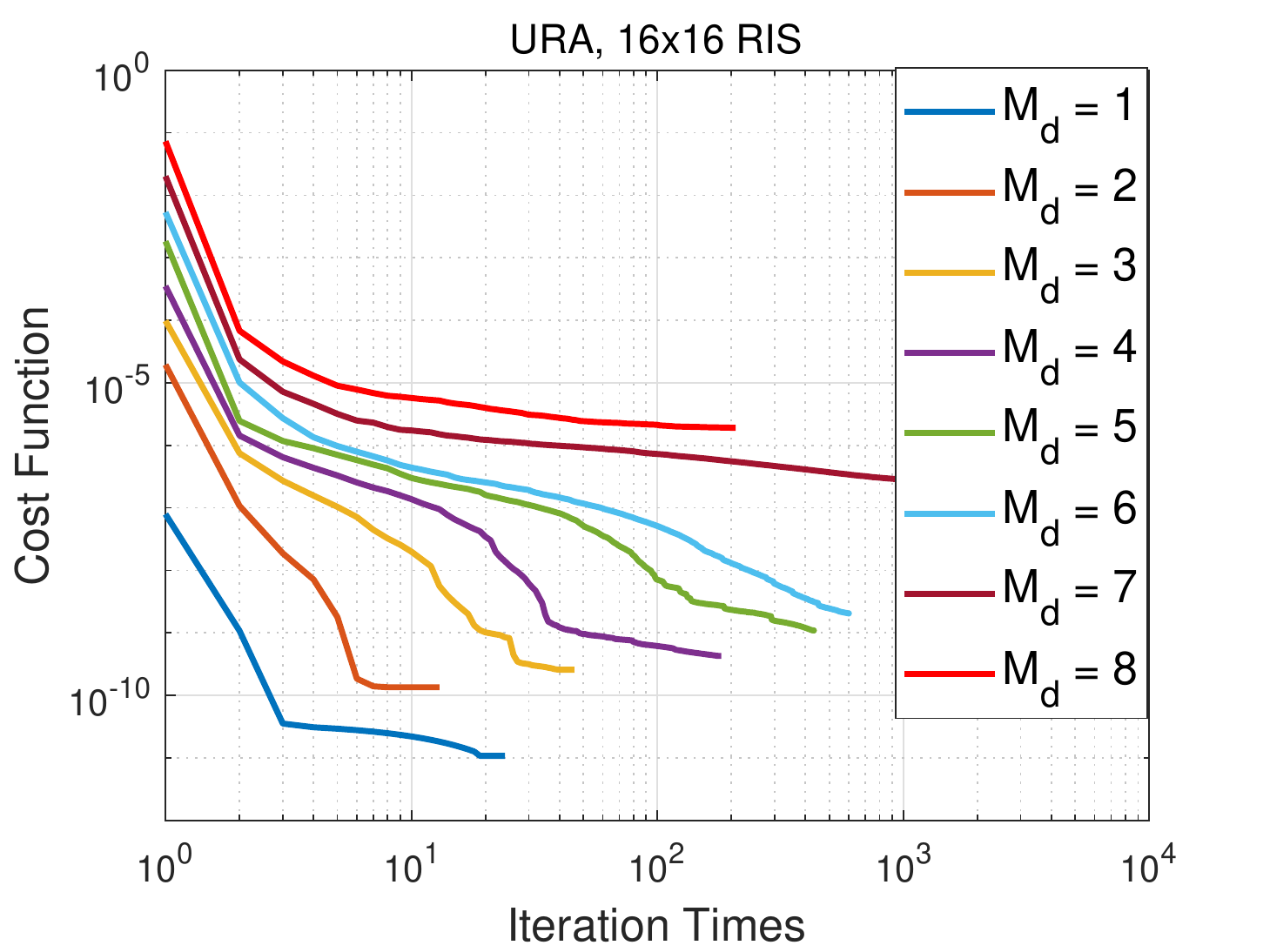,width= 2.3in}}
    \label{fig.URAc}
  \end{minipage}
  \caption{Convergence of cost function (\ref{costFuncv4}) when running Algorithm \ref{Algo.2} for URA shown in Fig. \ref{fig.BSstructure}(b).}
  \label{fig.URA}
\end{figure*}

For both the antenna arrays and the RIS, the inter-element spacing is $\frac{\lambda}{2}$, where $\lambda$ is the wavelength.
The antennas are $\frac{\lambda}{2}$ away from the plane of the RIS, and the distance between the center of the transmit and receive antenna arrays is $3\lambda$.
For wave number $k_{\lambda}=\frac{2\pi}{\lambda}$, the LOS channel between two adjacent points $A(x_1,y_1,z_1)$ and $B(x_2,y_2,z_2)$ is simulated as \cite{1552223}
\ben
h_{AB} = \sqrt{\beta_{AB}}e^{-jk_{\lambda}d_{AB}},
\label{equ.nearFieldChan}
\een
where
\ben
\beta_{AB} = \frac{G_l}{4}\left(\frac{1}{(k_{\lambda}d)^2}-\frac{1}{(k_{\lambda}d)^4}+\frac{1}{(k_{\lambda}d)^6}\right),
\een
and $d_{AB} = \sqrt{(x_1-x_2)^2+(y_1-y_2)^2+(z_1-z_2)^2}$.
Thus, the entries of $\Hbf_{B_rR}$, $\Hbf_{B_rB_t}$, and $\Hbf_{RB_t}$ are obtained from (\ref{equ.nearFieldChan}), and the entries of $\Hbf_{Ru}$, $\Hbf_{B_ru}$, $\Hbf_{dR}$, and $\Hbf_{dB_t}$ are of complex Gaussian with zero mean and unit variance.

In the following simulations, $\beta_{u,k},k=1,2,
\dots, K_u$ and $\beta_{d,k},k=1,2,\dots,K_d$ are generated according to the free-space path loss model
\ben
\beta_{u,k} = \left[\frac{\sqrt{G_l}\lambda}{4\pi d_{u,k}}\right]^2,\beta_{d,k} = \left[\frac{\sqrt{G_l}\lambda}{4\pi d_{d,k}}\right]^2
\een
where $d_{u,k}$ and $d_{d,k}$ are the distances from the BS to $k$th UL user and $k$ DL user, respectively. The maximum transmit power of BS is  $P_t=30$dBm, and the transmit power of UL users are $\sigma_{s_u}^2 = 10$dBm. We assume that $K_u = K_d = 3$ with $d_{u,k} = 100{\rm m}, k = 1,2,3$ and $d_{d,k}=500{\rm m}, k = 1,2,3$. The wavelength $\lambda = 0.125$m given the $2.4$GHz carrier frequency, and $G_l = 1$ as the BS are equipped with omnidirectional antennas. The noise power at the BS and users are set to be $\sigma_B^2=\sigma^2 = -95$dBm.
For the half-duplex mode, $\sigma_B^2=\sigma^2 = -98$dBm as the frequency bandwidth is evenly divided between the UL and DL transmission.

\begin{figure}[htb]
\centering
{\psfig{figure=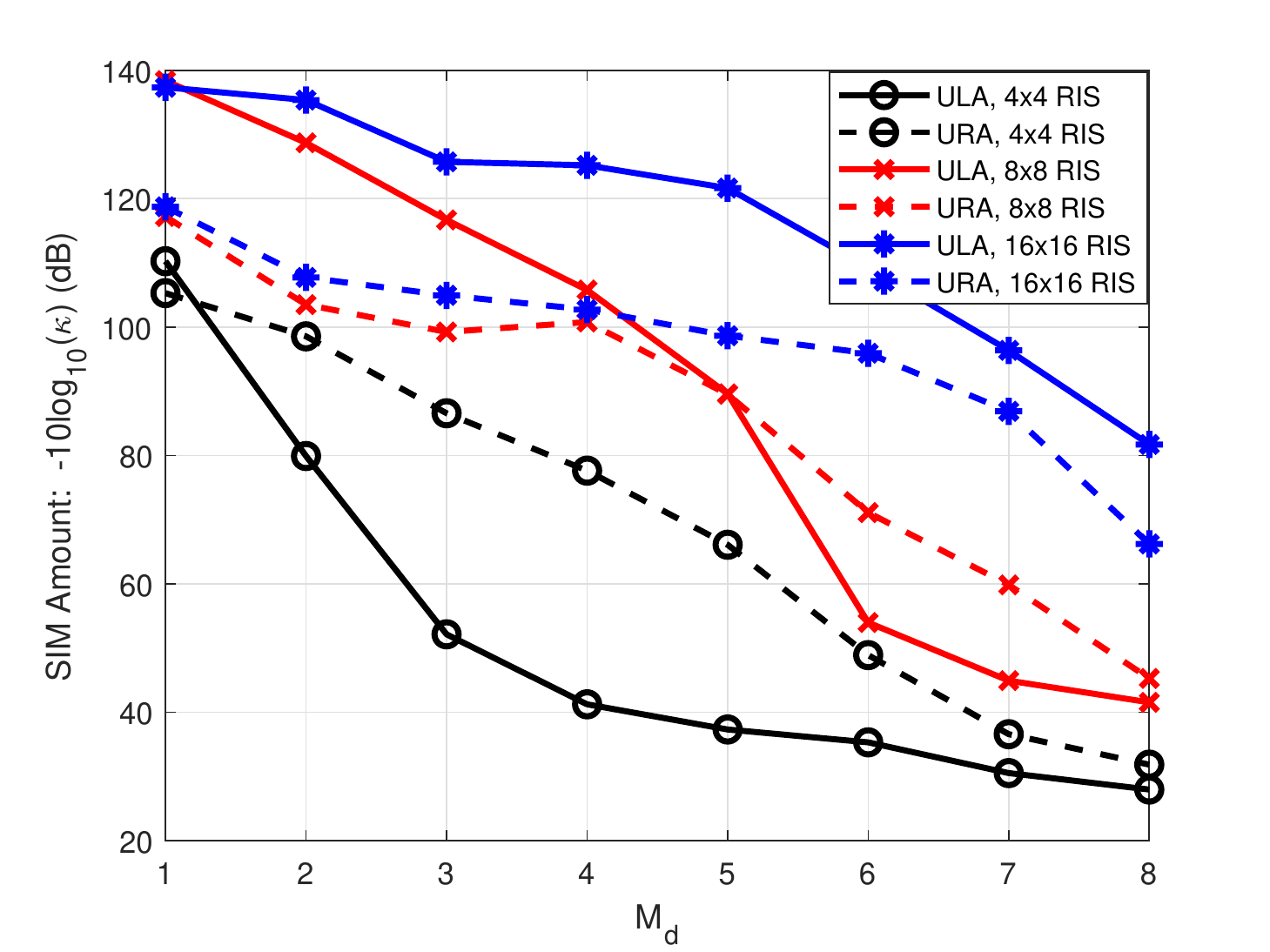,width= 3.5in}}
\caption{The SIM metric $\kappa$ vs. the number of DL effective antennas $M_d$.}
\label{fig.Pfance}
\end{figure}

\begin{figure}[htb]
\centering
{\psfig{figure=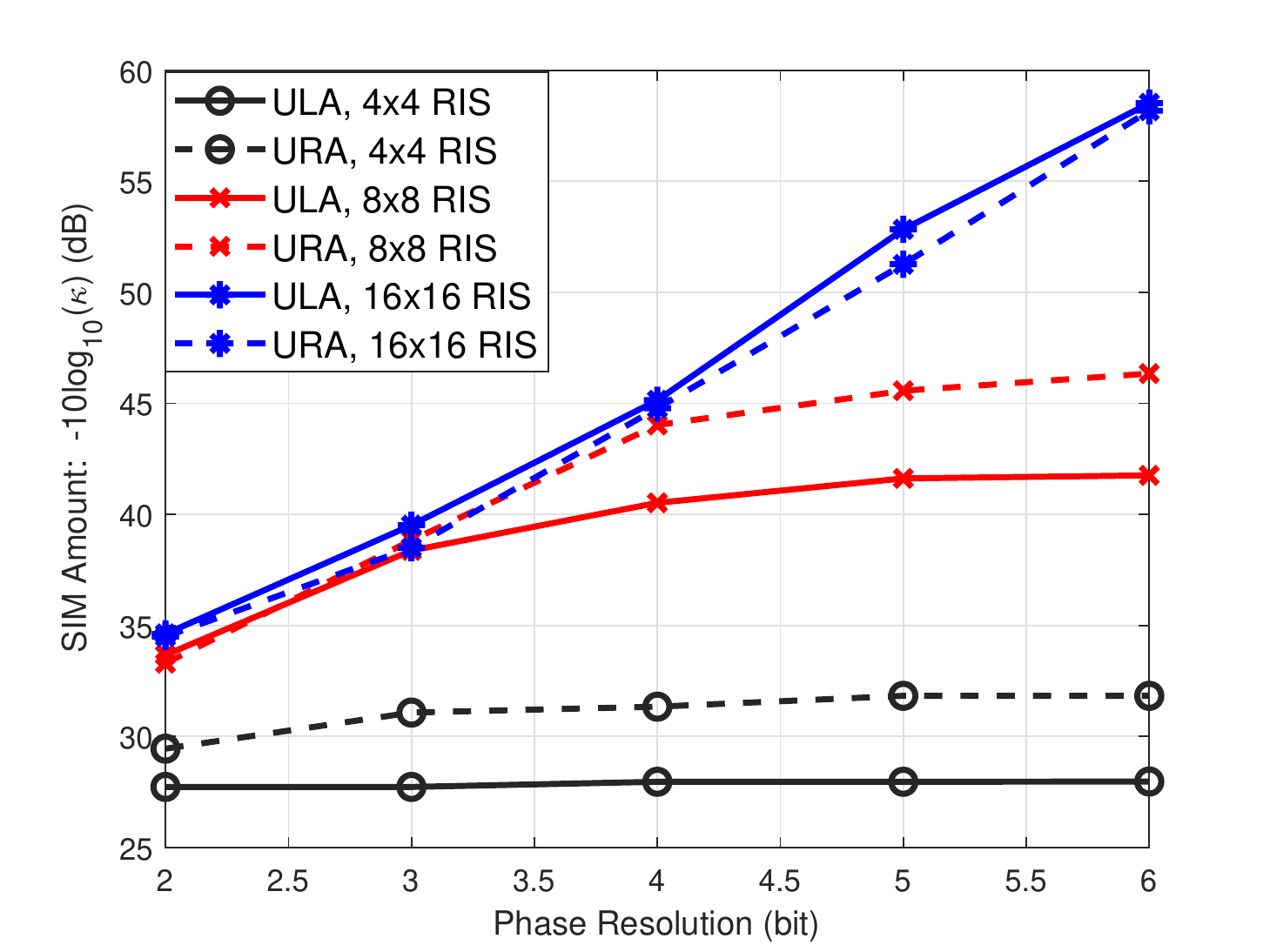,width= 3.5in}}
\caption{The SIM metric $\kappa$ vs. the phase resolution $b$.}
\label{fig.simvsb}
\end{figure}
\begin{figure*}[htb]
  \begin{minipage}[htb]{0.3\linewidth}
    \centering
    \subfigure[UL rates]{
    \psfig{figure=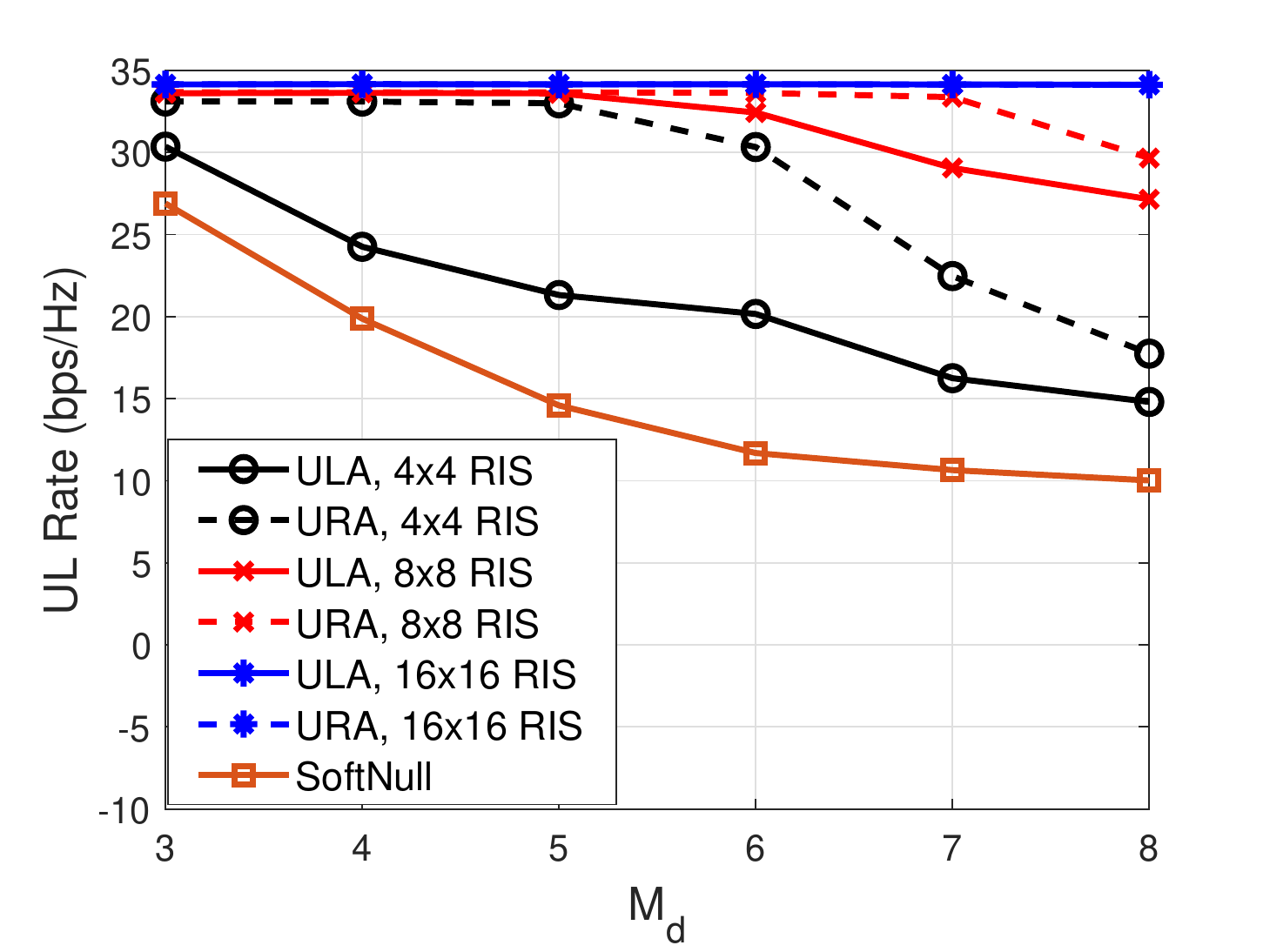,width= 2.3in}}
    \label{fig.ULRateMd}
  \end{minipage}%
  \begin{minipage}[htb]{0.3\linewidth}
    \centering
    \subfigure[DL rates]{
    \psfig{figure=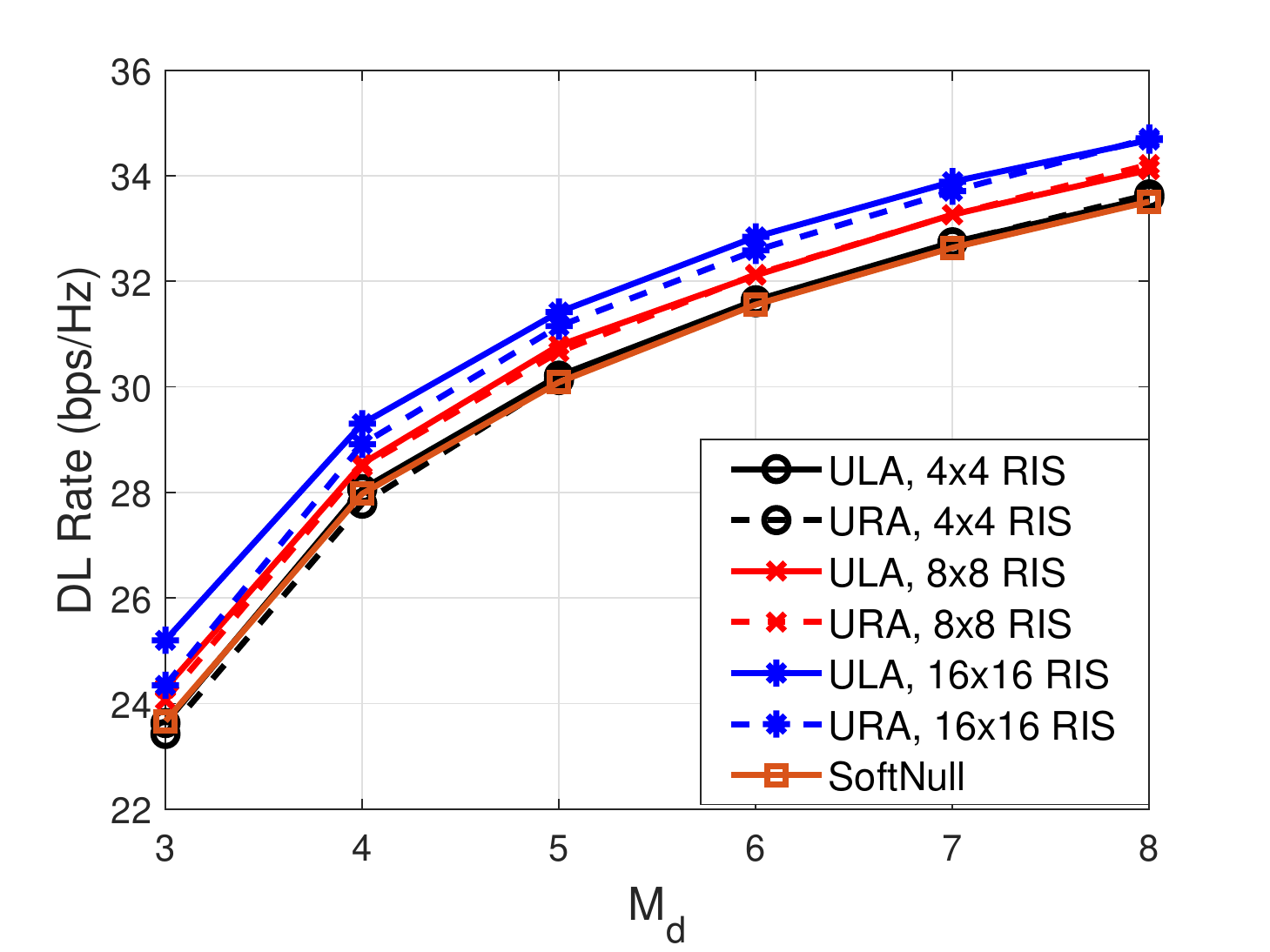,width= 2.3in}}
    \label{fig.DLRateMd}
  \end{minipage}
  \begin{minipage}[htb]{0.3\linewidth}
    \centering
    \subfigure[Sum-rates]{
    \psfig{figure=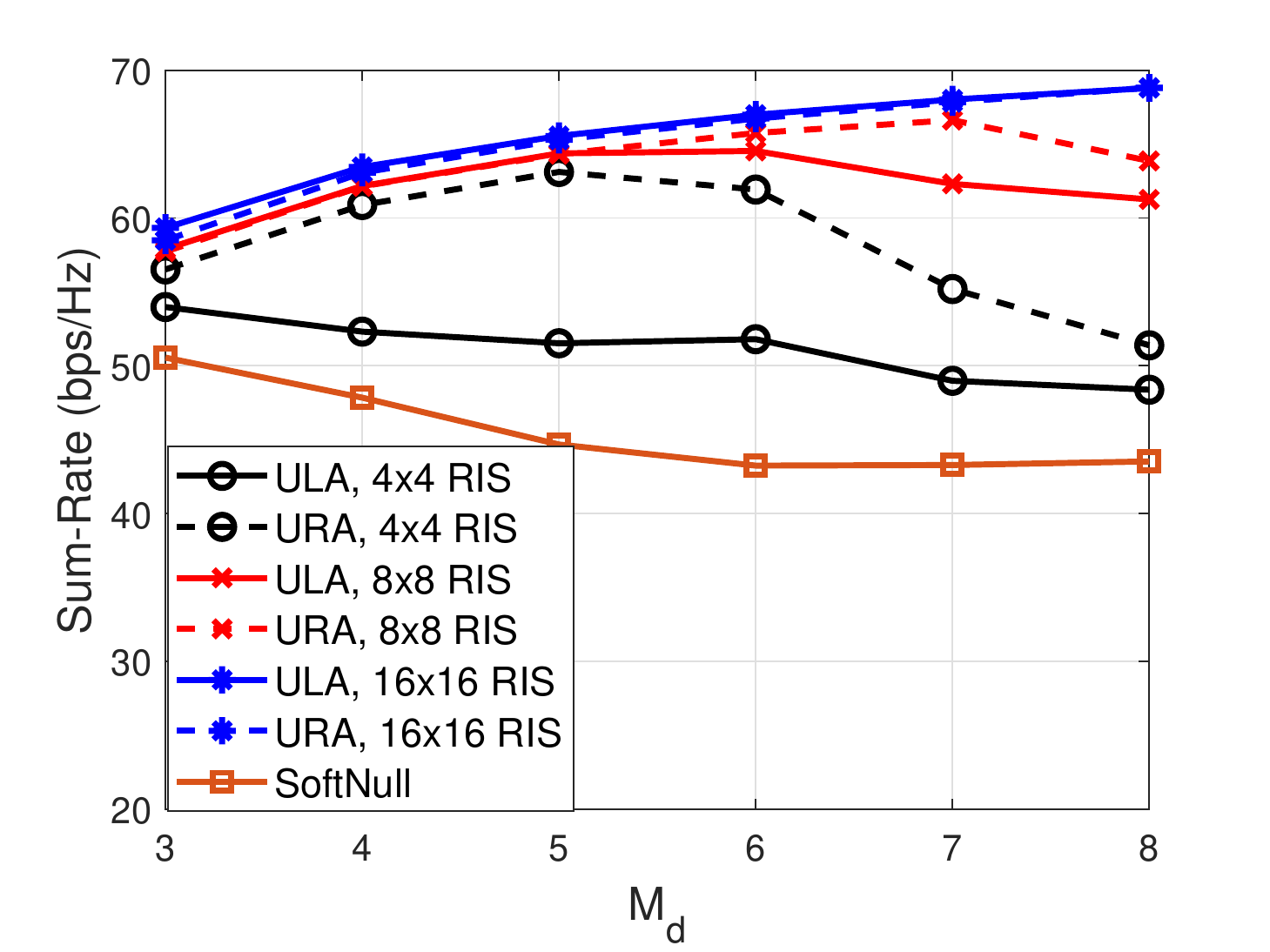,width= 2.3in}}
    \label{fig.SumRateMd}
  \end{minipage}
  \caption{The UL rate, the DL rate, and the sum-rate vs. the number of DL effective antennas $M_d$.}
  \label{fig.RateMd}
\end{figure*}
\begin{figure*}[htb]
  \begin{minipage}[htb]{0.3\linewidth}
    \centering
    \subfigure[UL rates]{
    \psfig{figure=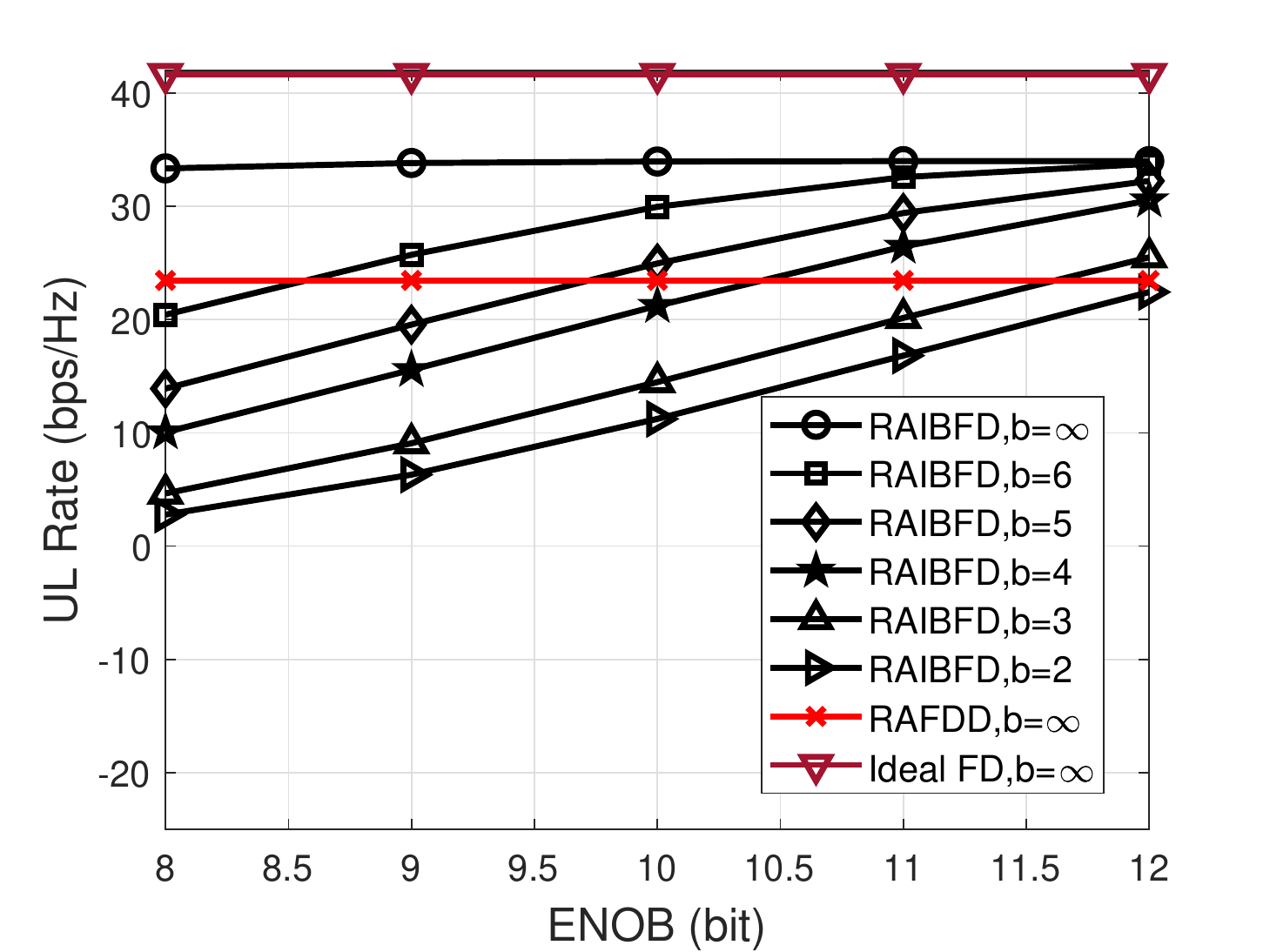,width= 2.3in}}
    \label{fig.ULRateADC}
  \end{minipage}%
  \begin{minipage}[htb]{0.3\linewidth}
    \centering
    \subfigure[DL rates]{
    \psfig{figure=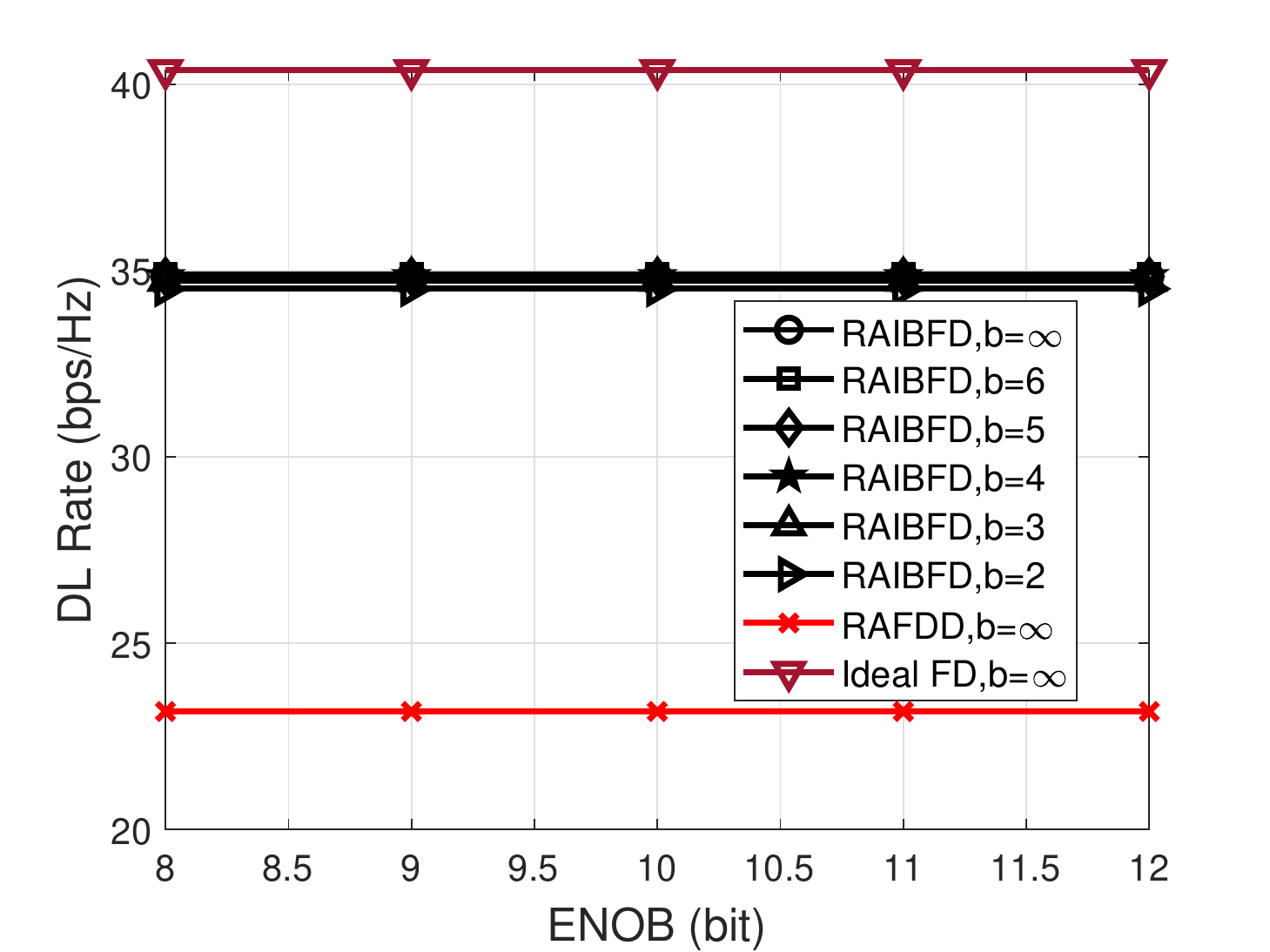,width= 2.3in}}
    \label{fig.DLRateADC}
  \end{minipage}
  \begin{minipage}[htb]{0.3\linewidth}
    \centering
    \subfigure[Sum-rates]{
    \psfig{figure=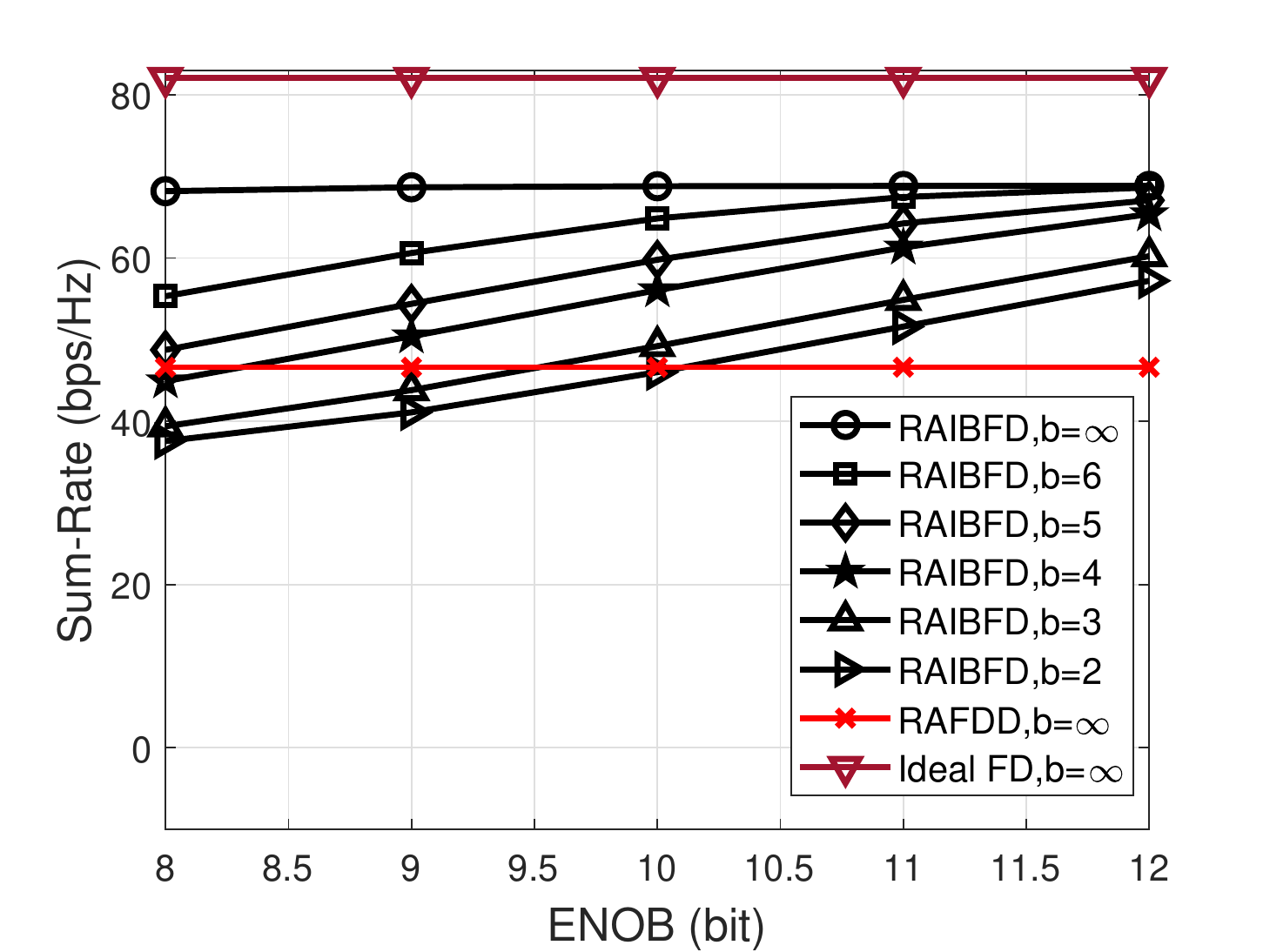,width= 2.3in}}
    \label{fig.SumRateADC}
  \end{minipage}
  \caption{The UL rate, the DL rate, and the sum-rate vs. the ${\sf ENOB}$ of the ADC.}
  \label{fig.RateADC}
\end{figure*}
In the first simulation, we study the convergence performance of the cost function (\ref{costFuncv4}) when running Algorithm \ref{Algo.1} for both ULA and URA scenarios where a $4 \times 4$ RIS, a $8 \times 8$ RIS, and a $16 \times 16$ RIS of $\infty$-bit resolution are deployed, respectively. Recall that $M_d$ is a design parameter to optimize. As shown in Fig. \ref{fig.ULA} and Fig. \ref{fig.URA},  as $M_d$ increases from $1$ to $8$ it becomes harder for cost function (\ref{costFuncv4}) to approach $0$. But employing more RIS elements can improve both the rate and the convergence. Using the SIM matrices $\Pbf_{\rm SIM}$'s and RIS response matrices $\Dbf$'s from the results of above convergence process and calculating the average SIM metric over receive antennas as
\ben
\kappa = \frac{1}{M_r}||(\Hbf_{B_rR}\Dbf\Hbf_{RB_t}+\Hbf_{B_rB_t})\Pbf_{\rm SIM}||_F^2,
\een
Fig. \ref{fig.Pfance} shows that both ULA and URA assisted by a RIS can suppress the self-interference effectively, where URA shows better SIM performance given a $4\times 4$ RIS and ULA  provide more SIM  than URA with a $16\times 16$. The ULA, assisted by a $16\times 16$ RIS, can suppress over $120$dB self-interference, which shows the promising potential of the proposed RIS-assisted SIM scheme. The non-negligible difference between the ULA and the URA suggests that the placement of the antennas can be optimized for meaningful improvement of the system performance, which can be an interesting topic for future research.
We study the SIM performance as the phase resolution varies from $2$ to $6$ in Fig. \ref{fig.simvsb}, where $M_d$ is set to be $8$. We can see that the $4\times 4$ RIS is not sensitive to the phase resolution, but higher phase resolution leads to larger SIM amount with RIS size of $8\times 8$ and $16\times 16$.

Fig. \ref{fig.RateMd} compares the UL rate, DL rate, and sum-rate performance of the proposed RIS-assisted in-band full-duplex (RAIBFD) versus SoftNull method proposed in \cite{SoftNull} as the number of DL effective antennas $M_d$ varies from $3$ to $8$. Setting the ADC bit resolution to be $b = 12$bit, we consider that a $4\times4$ RIS, a $8\times8$ RIS, and a $16\times16$ RIS of $\infty$-bit resolution are deployed for both ULA and URA scenarios. For $4\times 4$ RIS and $8\times 8$ RIS, although the increasing $M_d$ can improve the DL rate, it can cause more self-interference to the receive antennas and reduces the UL rate, thus, we see a trade-off between the UL and DL rate with respect to number of DL effective antennas $M_d$. But for $16 \times 16$ RIS, the proposed RAIBFD can completely eliminate the self-interference and achieve the best sum-rate performance for $M_d = 8$; in contrast, the sum-rate performance of the SoftNull method decreases as $M_d$ increases due to its limited capability of mitigating the self-interference.

The third example shows the UL rate, DL rate, and sum-rate performance of the proposed RAIBFD, the RIS-assisted FDD (RAFDD) proposed in \cite{RISFDD}, and the ideal full-duplex as the ENOB of ADC varies from $8$ to $12$ in the ULA scenario. For the RAFDD, UL carrier frequency and DL frequency are $f_u = 1760$MHz and $f_d = 1855$MHz. As shown in Fig. \ref{fig.RateADC}, the simulation of the proposed RAIBFD is conducted with the $16\times 16$ RIS with different phase resolution: $2$-bit, $3$-bit, ..., $6$-bit, and $\infty$-bit. Given the $16 \times 16$ RIS of $\infty$-bit resolution and $12$-bit ADCs, Fig. \ref{fig.RateADC}(c)  shows that proposed RAIBFD with RIS of $\infty$-resolution can achieve $83\%$ of sum-rate of the ideal full-duplex and has $48\%$ gain over the RAFDD. The sum-rate performance of the RAIBFD does degenerate for the RIS with finite bit resolution, but even with the $2$-bit $16\times 16$ RIS, the RAIBFD can still enjoy $23\%$ gain over the RAFDD as shown in Fig. \ref{fig.RateADC}(c) at ENOB $12$.

\begin{figure}[htb]
\centering
{\psfig{figure=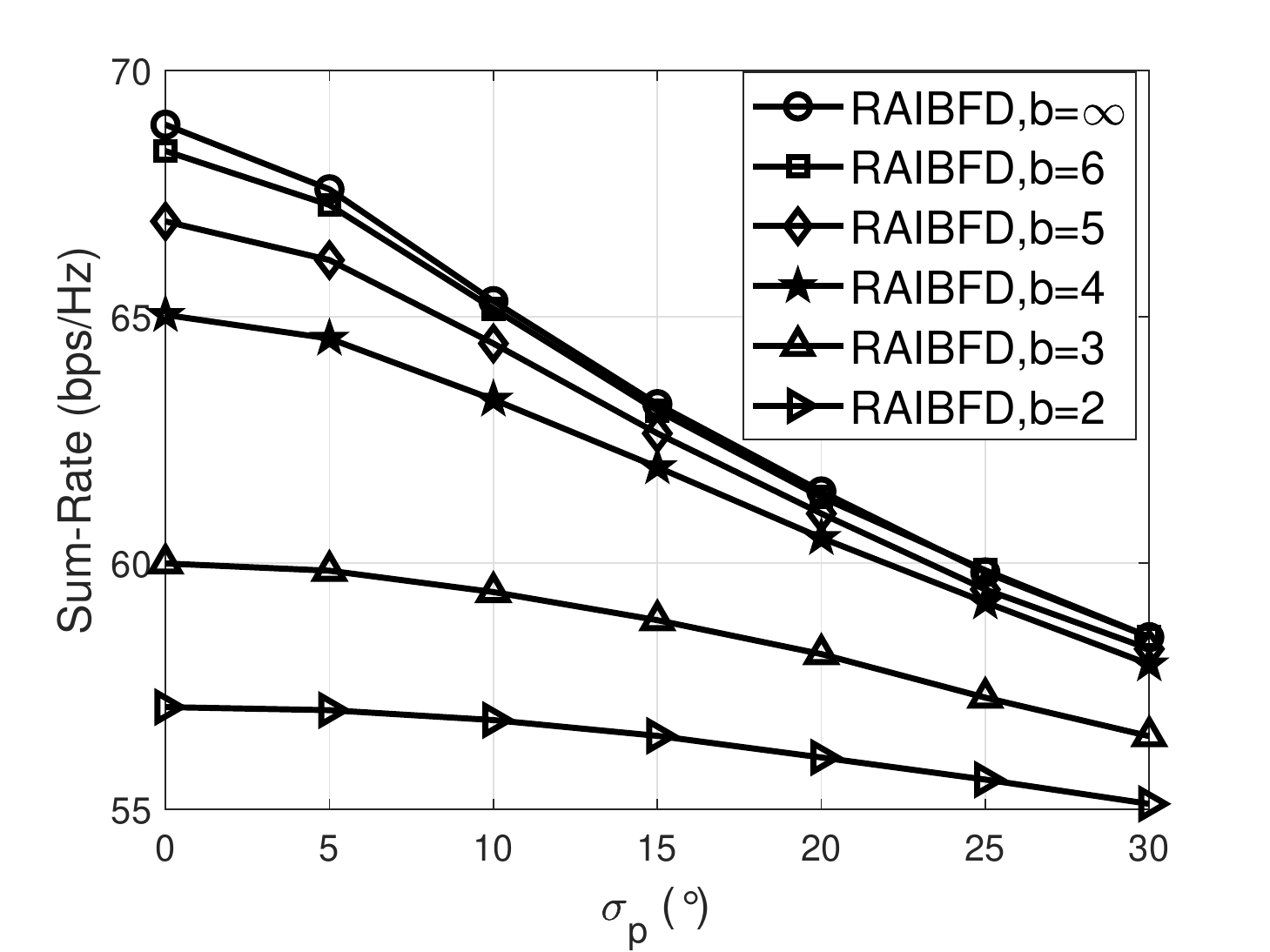,width= 3.5in}}
\caption{The sum-rate vs. $\sigma_p$ under different bit resolutions of a $16\times 16$ RIS.}
\label{fig.PhaseError}
\end{figure}
In the last simulation, we simulate the sum-rate performance of the proposed RAIBFD when the reflection coefficients of RIS suffer from phase deviations for manufacture imperfection. We consider a $16\times 16$ RIS deployed ULA with $M_d = 8$, and the ENOB of ADC is ${\sf ENOB} = 12$bit.
Assume that the phase deviations are uniformly distributed between $[-\sigma_p,\sigma_p]$.  Fig. \ref{fig.PhaseError} shows that the sum-rate performance of the proposed RAIBFD with RIS of $\infty$-bit and bit resolution ranging from $2$-bit to $6$-bit as $\sigma_P$ varies from $0^\circ$ to $30^\circ$.
The simulation shows $2$-bit and $3$-bit RIS shows robustness to the phase deviations, but
phase deviations uniformly distributed between $[-30^\circ,30^\circ]$ cause $15\%$ sum-rate loss to the proposed RAIBFD systems equipped with $6$-bit $16\times 16$ RIS, which suggests that calibrating the RIS reflection coefficients can be an interesting research topic in the future.

\section{Conclusion} \label{SEC6}
In this paper, we considered a in-band full-duplex wireless system where the base station (BS) aided by a RIS transmits signal to and receives signal from multiple users on the same frequency band simultaneously. Taking into account the  quantization noise of the analog-to-digital converters (ADCs), we propose to jointly design the DL precoding matrix and the RIS coefficients to mitigate the self-interference and further maximize the sum of UL and DL rates. The simulation results show the effectiveness of the proposed RAIBFD wireless system as its sum-rate outperforms that of the RIS-assisted half-duplex and the state-of-the-arts in-band full-duplex method.
\section*{Appendix: Derivation of (\ref{equ.propgderd})} \label{appendix}
First, we calculate $\nabla R_u(\dbf)$. Rewriting $\Hbf_u$ from (\ref{equ.Hu}) as
\ben
\Hbf_u =\Hbf_{B_rR}\Dbf\Hbf_{Ru}\Gammabf_u+\Hbf_{B_ru}\Gammabf_u,
\label{equ.Huv2}
\een
we have
\ben
{\partial \Hbf_u} = \Hbf_{B_rR}({\partial \Dbf})\Hbf_{Ru}\Gammabf_u,
\label{equ.HuDer}
\een
and
\ben
{\partial \Hbf_u^H} =\Gammabf_u^H{\Hbf}_{Ru}^H({\partial \Dbf^H})\Hbf_{B_rR}^H.
\label{equ.HuDerConjTrans}
\een
Owing to the formula $\partial({\rm log}_2({\rm det}(\Xbf)))=\frac{1}{{\rm ln}2}\tr(\Xbf^{-1}{\partial \Xbf})$, we can obtain from (\ref{equ.FduLSE}) that
\ben
{\partial R_u} = \frac{1}{{\rm ln}2}\frac{\sigma_{u}^2}{\sigma_B^2}\tr\left(\Fbf^{-1}_u\left(({\partial \Hbf_u})\Hbf_u^H+\Hbf_u({\partial \Hbf_u^H})\right)\right),
\label{equ.RuFddDer}
\een
where $\Fbf_u = \Ibf+\frac{\sigma_{u}^2}{\sigma_B^2}\Hbf_u\Hbf_u^H$. Inserting (\ref{equ.HuDer}) and (\ref{equ.HuDerConjTrans}) into (\ref{equ.RuFddDer}) yields
\ben
{\partial R_u} = \frac{1}{{\rm ln}2}\frac{\sigma_{u}^2}{\sigma_B^2}\tr\left(\Jbf_u(\partial \Dbf)+\Jbf_u^H(\partial \Dbf^H)\right),
\label{equ.RuderD}
\een
where
\ben
\Jbf_u = \Hbf_{Ru}\Gammabf_u\Hbf_u^H\Fbf^{-1}_u\Hbf_{B_rR}.
\label{equ.Ju}
\een
According to (\ref{equ.RuderD}), we can further obtain
\ben
\nabla R_u(\dbf) = \frac{1}{{\rm ln}2}\frac{\sigma_{u}^2}{\sigma_B^2}\diag\left(\Jbf_u^H\right).
\label{equ.RuDerPhi}
\een

Second, we calculate $\nabla R_d(\dbf)$. Inserting (\ref{equ.gammai}) into (\ref{equ.HdDLv2}), we have
\ben
\begin{split}
R_d =& K_d{\rm log}_2\left(\frac{P_t}{\sigma^2}+\tr\left((\Hbf_{d}\Hbf_{d}^H)^{-1}\right)\right) \\
&-\sum_{k=1}^{K_d}{\rm log}_2\left(K_d\tr\left(\Ebf_{k}(\Hbf_{d}\Hbf_{d}^H)^{-1}\right)\right),
\end{split}
\een
from which we can further obtain
\ben
\begin{split}
{\partial R_d} =& \frac{K_d}{{\rm ln}2}\frac{\tr\left({\partial(\Hbf_{d}\Hbf_{d}^H)^{-1}}\right)}{\frac{P_t}{\sigma^2}+\sum_{k=1}^{K_d}\gamma_{k}} \\
&-\frac{1}{{\rm ln}2}\sum_{k=1}^{K_d}\frac{\tr\left(\Ebf_{k}(\partial(\Hbf_{d}\Hbf_{d}^H)^{-1})\right)}{\gamma_{k}}.
\end{split}
\label{equ.RdFDDDer}
\een
As $\Hbf_d$ can be rewritten as
\ben
\Hbf_d = \Gammabf_d\Hbf_{dR}\Dbf\Hbf_{RB_t}+\Gammabf_d\Hbf_{dB_t},
\een
we have
\ben
{\partial \Hbf_d} = \Gammabf_d\Hbf_{dR}({\partial \Dbf})\Hbf_{RB_t},
\label{equ.HdDer}
\een
and
\ben
{\partial \Hbf_d^H} = \Hbf_{RB_t}^H({\partial \Dbf^H}){\Hbf}_{dR}^H\Gammabf_d^H.
\label{equ.HdDerTransp}
\een
Using the formula $\partial\left(\Xbf^{-1}\right) = -\Xbf^{-1}(\partial \Xbf)\Xbf^{-1}$ and denoting $\tilde{\Hbf}_d = (\Hbf_d\Hbf_d^H)^{-1}$, we can obtain
\ben
\begin{split}
{\partial \tilde{\Hbf}_d} =& -\tilde{\Hbf}_d\left({\partial (\Hbf_d\Hbf_d^H)}\right)\tilde{\Hbf}_d, \\
=&-\tilde{\Hbf}_d\Gammabf_d\Hbf_{dR}({\partial \Dbf})\Hbf_{RB_t}\Hbf_d^H\tilde{\Hbf}_d \\
&-\tilde{\Hbf}_d\Hbf_d\Hbf_{RB_t}^H({\partial \Dbf^H}){\Hbf}_{dR}^H\Gammabf_d^H\tilde{\Hbf}_d,
\end{split}
\label{equ.HdHdTranspInvDer}
\een
which leads (\ref{equ.RdFDDDer}) to be
\ben
\begin{split}
{\partial R_d} =& -\frac{K_d}{{\rm ln}2}\frac{\tr\left(\Fbf_d({\partial \Dbf})+\Fbf_d^H({\partial \Dbf^H})\right)}{\frac{P_t}{\sigma^2}+\sum_{k=1}^{K_d}\gamma_{k}} \\
&+ \frac{1}{{\rm ln}2}\sum_{k=1}^{K_d}\frac{\tr\left(\Jbf_{d,k}({\partial \Dbf})+\Jbf_{d,k}^H({\partial \Dbf^H})\right)}{\gamma_k},
\end{split}
\label{equ.RdDerD}
\een
with
\ben
\Fbf_d = \Hbf_{RB_t}\Hbf_d^H\tilde{\Hbf}_d\tilde{\Hbf}_d\Gammabf_d{\Hbf}_{dR},
\label{equ.Fd}
\een
and
\ben
\Jbf_{d,k} = \Hbf_{RB_t}\Hbf_d^H\tilde{\Hbf}_d\Ebf_k\tilde{\Hbf}_d\Gammabf_d{\Hbf}_{dR}.
\label{equ.Jdk}
\een
According to (\ref{equ.RdDerD}), we have
\ben
\nabla R_d(\dbf) = -\frac{K_d}{{\rm ln}2}\frac{\diag\left(\Fbf_d^H\right)}{\frac{P_t}{\sigma^2}+\sum_{k=1}^{K_d}\gamma_{k}} + \frac{1}{{\rm ln}2}\sum_{k=1}^{K_d}\frac{\diag\left(\Jbf_{d,k}^H\right)}{\gamma_k}.
\label{equ.RdFDDDerv2}
\een
From (\ref{equ.RuDerPhi}) and (\ref{equ.RdFDDDerv2}), we have $\nabla g(\dbf)$ as
\ben
\begin{split}
\nabla g(\dbf) =& -\left(\nabla R_u(\dbf) + \nabla R_d(\dbf)\right),\\
=& -\frac{1}{{\rm ln}2}\frac{\sigma_{u}^2}{\sigma_B^2}\diag\left(\Jbf_u^H\right)+\frac{K_d}{{\rm ln}2}\frac{\diag\left(\Fbf_d^H\right)}{\frac{P_t}{\sigma^2}+\sum_{k=1}^{K_d}\gamma_{k}} \\
&-\frac{1}{{\rm ln}2}\sum_{k=1}^{K_d}\frac{\diag\left(\Jbf_{d,k}^H\right)}{\gamma_k}.
\end{split}
\label{equ.sumRateDerPhi}
\een
\bibliographystyle{IEEEtran}
\bibliography{bib}
\end{document}